\documentclass[hyperref,final]{LMCS}

\def\doi{8 (2:11) 2012}
\lmcsheading%
{\doi}
{1--25}
{}
{}
{Oct.~\phantom07, 2011}
{Jun.~19, 2012}
{}

\usepackage[T1]{fontenc}
\usepackage{hyperref}
\usepackage{enumitem}
\usepackage{xspace}
\usepackage{ifthen}
\usepackage{algorithm}
\usepackage[noend]{algpseudocode}
\algrenewcommand{\algorithmiccomment}[1]{\hfill$\triangleright$ {\sffamily\itshape #1}}

\usepackage{mathtools}
\usepackage{gastex}

\usepackage{amsfonts,amstext,amsmath,amssymb,calc}

\usepackage{tikz}
\usetikzlibrary{matrix,arrows}

\usepackage{pdfsync}

\theoremstyle{plain}
\newtheorem{theorem}{Theorem}[section]
\newtheorem{lemma}[theorem]{Lemma}
\newtheorem{proposition}[theorem]{Proposition}
\newtheorem{corollary}[theorem]{Corollary}

\theoremstyle{definition}
\newtheorem{definition}{Definition}[section]
\newtheorem{example}{Example}[section]
\theoremstyle{remark}

\makeatletter \def\moverlay{\mathpalette\mov@rlay}
\newcommand\mov@rlay[2]{\leavevmode\vtop{%
    \baselineskip\z@skip
    \lineskiplimit-\maxdimen
    \ialign{\hfil$#1##$\hfil\cr#2\crcr}}}
\makeatother

\newcommand\minbasiselt{\raise1pt\hbox{$\scriptscriptstyle\protect\moverlay{\bullet\cr\bigcirc}$}\xspace}
\newcommand\basiselt{\raise1pt\hbox{$\scriptscriptstyle\bullet$}\xspace}

\newcommand{\partie}[1]{2^{#1}}

\newcommand\fyset[1]{\mathcal{F}(\vec #1)}
\newcommand\Pyset[1]{\mathcal{P}(\vec #1)}

\newcommand\eg{{\em e.g.}}

\newcommand\ie{{\em i.e.}}

\let\emptyset\varnothing
\let\savenabla\nabla
\renewcommand\nabla{\mathbin{\savenabla}}

\newcommand\trans{\xRightarrow{*}}
\newcommand\invar[2]{\ensuremath{{#1}\trans{#2}}}
\newcommand{\tup}[1]{\langle #1\rangle}

\newcommand{\dc}[1][\Default]{\def\Default{}%
  \ifthenelse{\equal{#1}{}}%
  {\mathop{\downarrow}}
  {\mathop{\downarrow}\nolimits_{#1}}\!}
\newcommand{\Dc}[1][\Default]{\def\Default{}%
  \ifthenelse{\equal{#1}{}}%
  {\mathop{\Downarrow}}
  {\mathop{\Downarrow}\nolimits_{#1}}\!}
\newcommand{\uc}[1][\Default]{\def\Default{}%
  \ifthenelse{\equal{#1}{}}%
  {\mathop{\uparrow}\!}
  {\mathop{\uparrow}\nolimits_{#1}\!}}

\newcommand\xdown[1]{\Dc[\vec #1]}
\newcommand\down[1]{\dc[{#1}]}
\newcommand\Pdown[1]{\dc[{\leq_#1}]}
\newcommand\filter[2]{\operatorname{Filter}(#1,\vec #2)}

\newcommand\vassz{\ensuremath{\text{VASS}_z}\xspace}
\newcommand\vass{\ensuremath{\text{VASS}}\xspace}
\newcommand\vasz{\ensuremath{\text{VAS}_z}\xspace}
\newcommand\vas{\ensuremath{\text{VAS}}\xspace}
\newcommand\az{\ensuremath{a_{\mathit z}}\xspace}

\newcommand\lequn{\ensuremath{\leq_1}\xspace}

\newcommand{\ru}[1]{\xrightarrow{#1}}

\newcommand{\Lim}{\operatorname{\mathsf{Lim}}}
\newcommand{\postset}{\operatorname{\mathsf{Reach}}}

\let\leq\leqslant
\let\geq\geqslant

\let\preceq\preccurlyeq

\renewcommand\vec[2][\Default]{\def\Default{}%
  \ifthenelse{\equal{#1}{}}%
  {\ensuremath{\boldsymbol{#2}}}%
  {\ensuremath{(\boldsymbol{#1},\boldsymbol{#2})}}%
}

\newcommand\poststar[1][\Default]{\def\Default{}%
  \ifthenelse{\equal{#1}{}}%
  {\ensuremath{\mathsf{Reach}}}%
  {\ensuremath{\mathsf{Reach}_{#1}}}}

\newcommand\ini{\ensuremath{\vec{x}_{\textit{in}}}\xspace}
\newcommand\iniprime{\ensuremath{\vec{x}'_{\textit{in}}}\xspace}
\newcommand\xini{\ini}

\newcommand\qini{\ensuremath{q_\mathit{in}}\xspace}

\newcommand\qf{\ensuremath{q_\mathit{f}}\xspace}

\newcommand\limra{\overset{n \to \infty}{\dashrightarrow}}

\newcommand\cover[2][\Default]{\def\Default{}%
  \ifthenelse{\equal{#1}{}}%
  {\ifthenelse{\equal{#2}{}}%
    {\ensuremath{\mathsf{Cover}}}%
    {\ensuremath{\mathsf{Cover}_{\vec #2}}}}%
  {\ifthenelse{\equal{#2}{}}%
    {\ensuremath{\mathsf{Cover}_{\vec #1}}}%
    {\ensuremath{\mathsf{Cover}_{\vec #1,#2}}}}}

\newcommand\xcover[2][\Default]{\def\Default{}%
  \ifthenelse{\equal{#1}{}}%
  {\ifthenelse{\equal{#2}{}}%
    {\ensuremath{\mathsf{Cover}}}%
    {\ensuremath{\mathsf{Cover}_{#2}}}}%
  {\ifthenelse{\equal{#2}{}}%
    {\ensuremath{\mathsf{Cover}_{#1}}}%
    {\ensuremath{\mathsf{Cover}_{#1,\vec #2}}}}}

\newcommand{\setN}{\ensuremath{\mathbb{N}}\xspace}
\newcommand{\setZ}{\ensuremath{\mathbb{Z}}\xspace}

\newcommand{\TT}{\ensuremath{\mathcal{T}}\xspace}
\newcommand{\NN}{\ensuremath{\mathcal{N}}\xspace}
\newcommand{\VV}{\ensuremath{\mathcal{V}}\xspace}

\newcommand{\embed}[1]{\ensuremath{\mathrel{{#1}^{*}}}}

\renewcommand{\figurename}{Fig.}

\title[Model Checking Vector Addition Systems with one zero-test]{Model Checking Vector Addition Systems with one zero-test\rsuper*}

\author[R\'emi Bonnet]{R\'emi Bonnet\rsuper a}
\address{{\lsuper{a,b,d}}LSV, ENS Cachan, CNRS \& INRIA, France}	
\email{firstname.lastname@lsv.ens-cachan.fr} 

\author[Alain Finkel]{Alain Finkel\rsuper b}
\address{{\lsuper{c,d}}LaBRI, Univ. Bordeaux \& CNRS, France}	
\email{firstname.lastname@labri.fr\vspace{-6 pt}} 
\author[J\'er\^ome Leroux]{J\'er\^ome Leroux\rsuper c}
\author[Marc Zeitoun]{Marc Zeitoun\rsuper d}

\thanks{{\lsuper{a,b,c,d}}Supported by the Agence Nationale de la Recherche,
    AVERISS (grant ANR-06-SETIN-001), AVERILES (grant
    ANR-05-RNTL-002), ANR 2010 BLAN 0202 01 FREC, and
    REACHARD-ANR-11-BS02-001.}

\titlecomment{{\lsuper*}Work based on the earlier extended abstracts \cite{FSTTCS:10,RP:11}.}

\begin{document}

\begin{abstract}
  We design a variation of the Karp-Miller algorithm to compute, in a
  forward manner, a finite representation of the cover (\ie, the
  downward closure of the reachability set) of a vector addition
  system with one zero-test. This algorithm yields decision procedures
  for several problems for these systems, open until now, such as
  place-boundedness or LTL model-checking. The proof techniques to
  handle the zero-test are based on two new notions of cover: the
  \emph{refined} and the \emph{filtered} cover. The refined cover is a
  hybrid between the reachability set and the classical cover. It
  inherits properties of the reachability set: equality of two refined
  covers is undecidable, even for usual Vector Addition Systems (with
  no zero-test), but the refined cover of a Vector Addition System is
  a recursive set. The second notion of cover, called the filtered
  cover, is the central tool of our algorithms. It inherits properties
  of the classical cover, and in particular, one can effectively
  compute a finite representation of this set, even for Vector
  Addition Systems with one zero-test.
\end{abstract}

\maketitle

\section{Introduction}

\paragraph{\bfseries Context: verifying properties of Vector Addition
  Systems.} Petri Nets, Vector Addition Systems
(VAS), and Vector Addition Systems with control States (VASS) are
equivalent well-known classes of counter systems for which the
reachability problem is decidable
\cite{Mayr:81,Kosaraju:82,LEROUX-POPL2011}, even if its complexity is
still open. On the other hand, testing equality of the
reachability sets of two such systems is
undecidable~\cite{Baker73,Hack:76}. For this reason, one cannot
compute a canonical finite representation of the reachability set that
would make it possible to test for equality of two reachability sets.
However, there is such an effective finite representation for the
\emph{cover}, a useful over-approximation of the reachability set
which is connected to various verification problems. Therefore, one
can decide not only the coverability problem (that is, membership to
the cover), but also whether two VAS have the same cover.

\smallskip

Vector Addition Systems are powerful models for the verification of
networks of identical finite-state machines communicating by
rendez-vous, with dynamic creation and destruction. Intuitively, a
global configuration of such a system is abstracted by nonnegative
counters, one for each possible location of the finite-state
machine. A counter value denotes the number of machines in the
corresponding location (see for instance \cite{Emerson:98}). Notice
that dynamic creation makes the number of processes, and therefore the
values of counters, possibly unbounded. For modeling client-server
systems where clients are identical finite-state machines, and the
server is another finite-state machine that can check that no process
is in a critical section, the VAS model is no longer
sufficient. Indeed, one must be able to check that a particular
counter is equal to zero, namely the one counting processes in the
critical section. This is a first practical motivation for adding to
VAS the ability to test a counter for 0.

Another reason to consider such a model is that it constitutes a first
step towards the verification of VAS equipped with a stack, a model
borrowing features both to pushdown automata and to VAS, and that
abstracts recursive programs manipulating constrained counters.
However, these systems are difficult to analyze. Abstracting away the
actual stack alphabet transforms the stack into a counter that can be
tested to zero. In this paper, we study verification problems for VAS
with one zero test.

\smallskip If one adds to VAS the ability to test at least two
counters for zero, one obtains a model equivalent to Minsky machines,
for which all nontrivial properties (in the sense of Rice's theorem)
of the language they recognize are undecidable, and many properties of
their behavior, such as reachability of a control state or
termination, are also undecidable. The study of VAS with \emph{a
  single} zero-test is recent, and only few results are known for this
model. Reinhardt~\cite{Reinhardt:08} has shown that the reachability
problem is decidable for VAS with one zero-test transition (as well as
for hierarchical zero-tests), and an alternate, simpler proof of this
result was recently given by the first author~\cite{MFCS:11}.  Abdulla
and Mayr have shown that the coverability problem is decidable
in~\cite{AbdullaM:09}, by using both the backward procedure of Well
Structured Transition Systems~\cite{AbdullaCJT:96} (see
\cite{Finkel&Schnoebelen:01} for a survey on Well Structured
Transition Systems), and the decidability of forward-reachability of
ordinary VASS as an oracle. The boundedness problem (whether the
reachability set is finite), the termination and the
reversal-boundedness problem (whether the counters can alternate
infinitely often between the increasing and the decreasing modes) are
all decidable by using a forward procedure, computing a finite, yet
\emph{incomplete}, Karp-Miller tree~\cite{Finkel&Sangnier:10}.

\paragraph{\bfseries LTL specifications}
Linear time temporal logic is a widely used specification logic, which
can express safety and liveness properties. Emerson~\cite{Emerson:98}
has designed an algorithm based on a covering graph to check LTL
properties on Well Structured Transition Systems, but which may not
terminate.  Esparza~\cite{Esparza:94,Esparza:98} has shown that LTL
specifications on the actions of a VAS is decidable (contrary to CTL)
and that LTL becomes undecidable if one adds state
predicates. Habermehl \cite{Habermehl:97} completed this proof by
showing EXPSPACE-completeness of LTL satisfiability, by generalizing
Rackoff's proof \cite{Rackoff:78}. These results have been unified
in~\cite{Blockelet&Schmitz:Model-Checking-Coverability-Graphs:2011:a}.

\paragraph{\bfseries Our contribution.} 
We give an algorithm for computing a \emph{finite representation} of the cover for a 
VAS with one zero-test. This result makes it possible to decide the 
place-boundedness problem, which is in general undecidable for VAS extensions 
(such as VAS with resets~\cite{Dufourd:98} or lossy counter machines,
\ie, lossy VAS with zero-test
transitions~\cite{Bouajjani:99,Mayr:03}).

Our proof first introduces a new notion of cover, called \emph{refined
  cover}, where the usual ordering on vectors is replaced by one that
insists on keeping equality on certain components. The refined cover
is a set hybrid between the reachability set and the classical cover.
We show that equality of two refined covers is undecidable, even for
usual VAS (with no zero-test). However, one can show that such a
refined cover is recursive for a VAS. We then introduce \emph{filtered
  covers}, the main technical tool of our algorithm. A filtered cover
is defined wrt.\ some specific values attached to some components. It
consists in retaining only these vectors from the reachability set
that agree with these values, before taking the usual downward
closure. By transferring decidability results from refined covers to
filtered covers, we are able to compute a finite representation of any
filtered cover. We use this representation to propose an algorithm
\emph{\`a la} Karp and Miller, which builds a tree to compute the cover
of a \vas with one zero-test. This allows us to obtain new
decidability results for such systems, namely for the classical
problems of place-boundedness. Finally, we show that the repeated
control state reachability for vector addition systems with states and
one zero-test is decidable, as well as LTL model-checking, by reducing
these problems to the reachability problem. Note that, for VASS (with
no zero-test), both problems can be reduced to the computation of the
cover set. We do not know whether there is such a reduction between
the corresponding problems for VASS with one zero test, and we leave
it as an open problem.

\smallskip Thus, this work can be viewed as a contribution to
understanding the limits of decidability, taking into account two
parameters: the models (VAS and VAS with one zero-test) and the
problems (reachability, cover, refined and filtered cover).

\paragraph{\bfseries The difficulty.}
\label{sec:difficulty}
The central problem is to compute the cover of a VAS \emph{with one
  zero-test}. Let us explain why the usual Karp-Miller algorithm
is not sufficient for that purpose. A crucial property of VAS used by
this algorithm is \emph{monotony}: actions fireable from a
state are still fireable from any larger state. This property is
clearly broken by the zero-test.

\smallskip
A natural idea appearing in~\cite{Finkel&Sangnier:10} is to adapt the
classical Karp-Miller construction~\cite{Karp&Miller:69}, first
building the Karp-Miller tree, but \emph{without} firing the zero
test. To continue the construction after this first stage, we need to
fire the zero test from the leaves of the Karp-Miller tree carrying a
$0$ value on the component that is tested to~$0$. The problem is that
accelerations performed while building the Karp-Miller tree may have
produced, on this component in the label of such a leaf, an $\omega$ value
that represents arbitrarily large values, and that abstracts actual
values. For this reason, one may not be able to determine if the zero
test succeeds or not. We therefore want a more accurate information
for the labeling of the leaves, for the component tested to~$0$. This
is what the filtered cover actually captures.

\smallskip To be more precise, let us illustrate this difficulty with
some short examples (assuming basic knowledge on VAS/VASS, see
Sec.~\ref{sec:vas}/\ref{sec:rcsrp}). The Karp-Miller algorithm
\cite{Karp&Miller:69,Finkel:minimal-coverability-graph-Petri:1993:a}
computes a finite representation of the cover of a VASS, \ie, the
downward closure of its reachability set (for the usual ordering over
$\setN^d$, where $d$ is the dimension of the VASS). It builds a finite
tree, whose nodes are labeled by elements of $(\setN\cup\{\omega\})^d$, where
intuitively $\omega$ represents arbitrary large values. At the end of the
algorithm, the cover is exactly the set of vectors of $\setN^d$ belonging
to the downward closure of the set of labels. The tree is obtained by
unwinding the system, and by performing acceleration when possible, in
order to guarantee termination: if one finds two nodes on the same
branch, such that the lowest one in the branch is labeled by a greater
element, one replaces by $\omega$ all components that have grown (this
captures the iteration of the firing sequence between the two nodes,
and this is where monotony is~used). We aim at generalizing this
algorithm for VASS with one zero-test.

\smallskip As a first example, consider in dimension~1 the two VASS
with one zero-test represented in \figurename~\ref{fig:km1}.
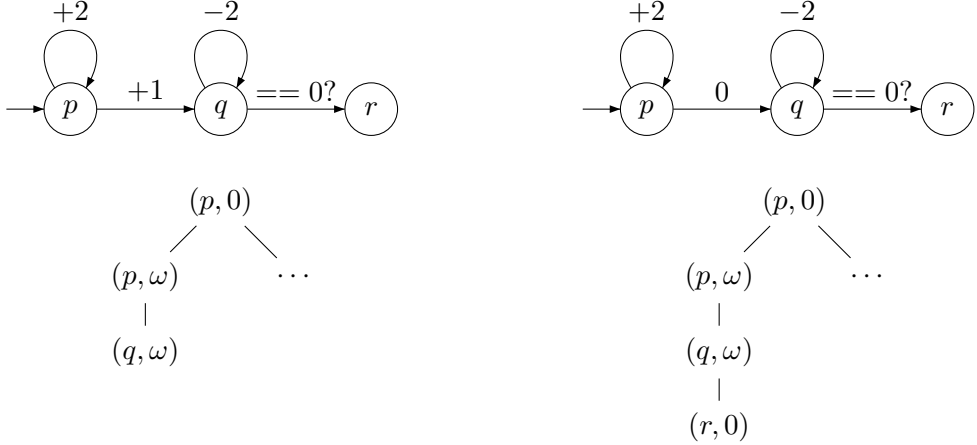
\begin{figure}[ht]
  \centering
  \scalebox{1}{%
    \begin{picture}(60,14)(-30,-2)
      \gasset{Nw=7,Nh=7,ExtNL=n,NLangle=0,NLdist=0,loopdiam=7}
      \node[Nmarks=i](p)(-20,0){$p$}
      \node(q)(0,0){$q$}
      \node(r)(20,0){$r$}
      \drawloop(p){$+2$}
      \drawloop(q){$-2$}
      \drawedge(p,q){$+1$}
      \drawedge(q,r){$==0?$}
    \end{picture}
    \qquad\qquad
    \begin{picture}(60,14)(-30,-2)
      \gasset{Nw=7,Nh=7,ExtNL=n,NLangle=0,NLdist=0,loopdiam=7}
      \node[Nmarks=i](p)(-20,0){$p$}
      \node(q)(0,0){$q$}
      \node(r)(20,0){$r$}
      \drawloop(q){$-2$}
      \drawloop(p){$+2$}
      \drawedge(p,q){$0$}
      \drawedge(q,r){$==0?$}
    \end{picture}}

  \vspace*{3ex}\hspace*{2ex}
  \scalebox{1}{%
    \begin{picture}(60,37)(-20,-32)
      \gasset{Nframe=n,AHnb=0,Nh=7,Nadjust=w,ExtNL=n,NLangle=0}
      \node(0)(0,0){$(p,0)$}
      \node(00)(-10,-10){$(p,\omega)$}
      \node(01)(10,-10){$\cdots$}
      \node(000)(-10,-20){$(q,\omega)$}
      \node[Nframe=n](0000)(-10,-30){}
      \drawedge(0,00){}
      \drawedge(0,01){}
      \drawedge(00,000){}
    \end{picture}
    \begin{picture}(60,37)(-35,-32)
      \gasset{Nframe=n,AHnb=0,Nh=7,Nadjust=w,ExtNL=n,NLangle=0}
      \node(0)(0,0){$(p,0)$}
      \node(00)(-10,-10){$(p,\omega)$}
      \node(01)(10,-10){$\cdots$}
      \node(000)(-10,-20){$(q,\omega)$}
      \node(0000)(-10,-30){$(r,0)$}
      \drawedge(0,00){}
      \drawedge(0,01){}
      \drawedge(00,000){}
      \drawedge(000,0000){}
    \end{picture}}
  \caption{Two VASS with one zero-test, and their Karp-Miller trees}
  \label{fig:km1}
\end{figure}
They only differ by the transition from $p$ to $q$. The transition
from $q$ to $r$ is the zero-test, fireable only when the counter is 0,
and which does not affect the counter. Starting from the initial state
$(p,0)$ and firing the loop from $p$ to itself, the algorithm first
computes as left child of the root a node labeled $(p,2)$, which then
gets accelerated as $(p,\omega)$. Then, firing the transition from $p$ to
$q$ yields the node $(q,\omega)$. Now, the zero-test is not fireable in the
first case, while it is fireable in the second case. Therefore, the
Karp-Miller trees we want to compute should differ (see
\figurename~\ref{fig:km1}, which shows two such partial Karp-Miller
trees). However, this cannot be detected with the information
available on the branch from $(p,0)$ to $(q,\omega)$, because this
information is identical for both systems: it consists of the nodes
$(p,0)$, $(p,\omega)$, $(q,\omega)$. This example illustrates the fact that the
$\omega$ component, in $(q,\omega)$, hides the actual reachable values, and
therefore also hides the ability or inability to fire the zero-test.

\smallskip The next example (\figurename~\ref{fig:km2}) is in
dimension 2. The zero-test occurs on the first component. It shows
that even if one could determine when to fire the zero-test, one might
be unable to compute the relevant node labeling using only information
provided by classical Karp-Miller trees.
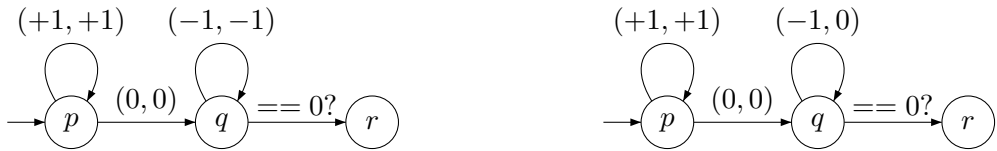
\begin{figure}[h]
  \centering
  \scalebox{1}{%
    \begin{picture}(45,22)(-20,-4)
      \gasset{Nw=7,Nh=7,ExtNL=n,NLangle=0,NLdist=0,loopdiam=7}
      \node[Nmarks=i](p)(-20,0){$p$}
      \node(q)(0,0){$q$}
      \node(r)(20,0){$r$}
      \drawloop(p){$(+1,+1)$}
      \drawloop(q){$(-1,-1)$}
      \drawedge(p,q){$(0,0)$}
      \drawedge(q,r){$==0?$}
    \end{picture}
    \qquad\qquad\qquad
    \begin{picture}(45,22)(-30,-4)
      \gasset{Nw=7,Nh=7,ExtNL=n,NLangle=0,NLdist=0,loopdiam=7}
      \node[Nmarks=i](p)(-20,0){$p$}
      \node(q)(0,0){$q$}
      \node(r)(20,0){$r$}
      \drawloop(p){$(+1,+1)$}
      \drawedge(p,q){$(0,0)$}
      \drawloop(q){$(-1,0)$}
      \drawedge(q,r){$==0?$}
    \end{picture}}
  \caption{Two VASS with one zero-test}
  \label{fig:km2}
\end{figure}
Indeed, the Karp-Miller trees for both systems before firing the
zero-test are identical. However, firing the zero-test from $(q,\omega,\omega)$
should produce a node labeled $(r,0,0)$ in the first case, and
$(r,0,\omega)$ in the second one.  Here, $\omega$ values in $(q,\omega,\omega)$ hide
relevant relationships between components (namely, that both
components remain equal in the first system).

\paragraph{\bfseries The schema of our proof.}
\begin{enumerate}[leftmargin=*]
\item We start in Section~\ref{sec:set-limits-reachable} with usual
  VAS: we extend the decidability of the reachability problem for VAS,
  by proving that the set $\Lim\poststar{}$ of \emph{limits} of
  sequences of reachable states is also recursive. This set
  $\Lim\poststar{}$ contains the reachability set, and captures more
  information, in general. Actually, it is more sophisticated than
  both the cover and the reachability set: it allows one to know
  whether an element in $(\setN\cup\{\omega\})^d$ is a reachable state or if
    it is the limit of a sequence of reachable states. This
  information is not given by the reachability set, neither by the
  cover (using the pointwise ordering over $(\setN\cup\{\omega\})^d$, and the
    natural ordering over $\setN\cup\{\omega\}$: $n\leq\omega$ for all $n$). The proof
  carries on by using Higman's Lemma, using a nontrivial~ordering.

\item In Section~\ref{sec:filtered-covers},  we refine the definition of 
  cover in which the first component of the vectors has now to be known 
  exactly (and not only bounded by some maximal value).
  We prove that, for VAS, the fact that $\Lim\poststar{}$ is recursive
  implies that one can \emph{compute} the finite basis of this filtered cover.

\item In Section \ref{sec:alg}, we compute the finite basis of the
  cover of a VAS with one zero-test by using a variation of the
  Karp-Miller algorithm that uses the previously defined filtered
  covers in order to convey enough information to go through the
  zero-test.

\item We add control states to our VAS with one zero-test in Section
  \ref{sec:rcsrp}, and we show that one can detect reachable
  increasing loops on a given control state, by reducing this problem
  to the reachability problem for VASS with one zero-test, a decidable
  problem~\cite{Reinhardt:08,MFCS:11}. This allows us to decide
  repeated control state reachability. We also note that this makes it
  possible to solve model checking against LTL or $\omega$-regular
  specifications. However, contrary to the situation without any
  zero-test, this is obtained by reducing this problem to the
  reachability problem, and not to the computation of the
  cover. Whether a reduction to this simpler problem exists is left
  open.
\end{enumerate}

\section{Preliminaries}
\label{sec:prelim}

\paragraph{\bfseries Words.} We denote by $A^*$ the set of finite
words over $A$. A word $u \in A^*$ is written $a_1 \cdots a_n$, with $a_i \in
A$. The concatenation of two words $u$ and $v$ is simply written $u v$
and the empty word is denoted $\varepsilon$, with $\varepsilon a = a \varepsilon = a$. We let
$A^+=A^*\setminus\{\varepsilon\}$ be the set of nonempty words.

\paragraph{\bfseries Orderings.} An \emph{ordering} $\preceq$ on a set~$X$
is a reflexive, transitive and antisymmetric binary relation
over $X$. Given $x,y\in X$, we write $x\prec y$ for $x\preceq y$ and $x\neq y$.
For $Y \subseteq X$, let $$\down{\preceq} Y = \{ x \in X\mid \exists y \in Y,\; x \preceq y\}$$ denote
the \emph{downward closure} of $Y$ with respect to $\preceq$.  The set $Y$
is said \emph{downward closed} if~$Y=\down{\preceq} Y$.  When working in
$\setN^d$ or $\setN^d_\omega$ with the usual ordering $\leq$ (see below), we
shorten the corresponding downward closure operator $\dc[\leq]\,$
as~${\dc}\,$.  Symmetrically, the \emph{upward closure} of $Y \subseteq X$,
denoted $\uc[\preceq] Y$ is defined by
\begin{equation*}
  \uc[\preceq] Y = \{ x \in X \mid \exists y \in Y,~ y \preceq x \}.
\end{equation*}
The set $Y$ is said to be \emph{upward closed} if $\uc[\preceq] Y = Y$.

\paragraph{\bfseries Vectors.}
For $d\geq1$, we write any vector $\vec x\in X^d$ as $\vec x=(\vec
x(1),\ldots,\vec x(d))$, with $\vec x(i)\in X$. Given an ordering $\preceq$ over
$X$, the \emph{pointwise ordering} over $X^d$, still denoted $\preceq$, is
defined by $\vec x\preceq\vec y$ if $\vec x(i)\preceq\vec y(i)$ for all~$i$.  For
$X=\setN$, we let~$\vec 0$ be the vector whose components are all~0, and
we say that $\vec x$ is \emph{nonnegative} if $\vec x\geq\vec 0$. For
$i\in\{1,\ldots,d\}$, we let $\vec{e}_i$ be the vector such that
$\vec{e}_i(i)=1$ and $\vec{e}_i(k)=0$ if~$k\neq i$.

\paragraph{\bfseries Limits in $\setN^d_\omega$.} We introduce an element
$\omega\not\in\setN$ and the set $\setN_\omega=\setN\cup\{\omega\}$. A sequence $(\ell_n)_{n\geq0}$ (also
written $(\ell_n)_n$) of elements of $\setN_\omega$ \emph{converges to $\ell\in\setN_\omega$},
if either it is ultimately constant with value $\ell$, or its subsequence
of integer values is infinite, tends to infinity, and~$\ell=\omega$.  We then
say that $\ell$ is \emph{the} \emph{limit} of $(\ell_n)_n$, noted
$\lim_n\ell_n=\ell$, or $\ell_n \limra \ell$.  A sequence $(\vec{x}_n)_n$ of
vectors of $\setN^d_\omega$ has limit $\vec x\in\setN^d_\omega$, noted
$\lim_n\vec{x}_n=\vec x$, if $\lim_n\vec{x}_n(i)=\vec x(i)$ for all
$i\in\{1,\ldots,d\}$.

\smallskip For $\vec M\subseteq\setN^d_\omega$, let $\Lim \vec M$ be the set of limits
of sequences of elements of $\vec M$. Notice that
\begin{equation}
  \label{eq:1}
  \vec M\subseteq\Lim \vec M,
\end{equation}
and
\begin{equation}
  \label{eq:2}
  \text{if } \vec M\subseteq\setN^d, \text{ then } \vec M=\setN^d\cap\Lim \vec M.
\end{equation}
Topologically speaking, $\Lim \vec M$ is the least limit closed set
containing $\vec M$. It is called the \emph{limit closure} of~$\vec
M$. The set $\vec M$ is said to be \emph{limit closed} if $\vec M=\Lim
\vec M$.

\paragraph{\bfseries Downward closed sets of\/ $\setN^d$ and $\setN^d_\omega$.}
Given an ordered set, one may under suitable hypotheses construct a
topological completion of this set, to recover a \emph{finite
  description} of its downward closed
subsets~\cite{Finkel&Goubault-Larrecq:09:a,Finkel&Goubault-Larrecq:09:b}. The
completion of $(\setN^d,\leq)$ is $(\setN_\omega^d,\leq)$ where we extend the ordering
$\leq$ over $\setN$ by $n\leq\omega$ for all $n\in\setN_\omega$.

\smallskip
A \emph{basis} of a set $\vec D\subseteq\setN^d_\omega$ is a \emph{finite} set $\vec
B\subseteq\setN_\omega^d$ such that
\begin{equation}
\label{eq:3}
\Lim \vec D=\down{} \vec B.
\end{equation}
Such a set $\vec B$ is a finite representation of $\Lim \vec D$. One
verifies that the maximal elements of any basis~$\vec B$ of $\vec D$
still form a basis, which only depends on $\vec D$. It is minimal for
inclusion among all bases, and is called \emph{the minimal basis} of
$\vec D$. Of course, not all sets admit a
basis. By~\cite{Finkel&Goubault-Larrecq:09:a,Finkel&Goubault-Larrecq:09:b},
any downward closed set~$\vec D\subseteq\setN^d$ admits a basis. This extends to
any downward closed set $\vec D$ of $\setN^d_\omega$. Indeed, one can check
that
\begin{equation}
  \label{eq:4}
  \Lim\vec D=\Lim(\vec{D}\cap\setN^d),
\end{equation}
so that a basis $\vec B$ of the downward closed set $\vec{D}\cap\setN^d$
satisfies $\Lim \vec D=\down{} \vec B$. Note that conversely, if $\vec
B\subseteq\setN^d_\omega$ is \emph{finite}, then $\down{}\vec B$ is limit~closed (this
may fail if $\vec B$ is infinite). Finally, the limit and downward
closure operators commute:
\begin{equation}
  \label{eq:5}
  \dc\Lim\vec M=\Lim\dc\vec M
\end{equation}

\paragraph{\bfseries Upward closed sets.} 
If~$\preceq$ is a well ordering over $X$ (see
Sec.~\ref{sec:set-limits-reachable} page~\pageref{wqo}), then for any
upward closed set $Y \subseteq X$, there exists a finite set $B \subseteq Y$ such that
$Y = \uc[\preceq] B$. Such a set is again called a basis (as for downward
sets, but there will be no ambiguity). Observe that contrary to the
case of downward closed sets, no topological completion is needed
here.

\begin{example}      
  Consider the set $\vec D=\bigl\{(x,y)\in\setN^2\mid x\leq3\lor
  y\leq1\bigr\}\cup\bigl\{(4,2),(4,3),(5,2)\bigr\}$, which is downward
  closed. It is represented by the greyed grayed area in
  \figurename~\ref{fig:exmaple-limit}. Its limit closure~is
  \begin{math}
    \Lim\vec D=\vec D\cup\bigl(\{0,1,2,3\}\times\{\omega\}\bigr)\cup\{\omega\}\times\{0,1\}.
  \end{math}
  A non-minimal basis of $\vec D$ is $(\Lim\vec D\setminus\vec
  D)\cup\{(4,3),(5,2)\}$, shown with dots \basiselt and \minbasiselt in
  \figurename~\ref{fig:exmaple-limit}, where elements involving~$\omega$
  fall beyond the grid. Its minimal basis is
  $\{(3,\omega),(4,3),(5,2),(\omega,1)\}$ (circled~\minbasiselt in
  \figurename~\ref{fig:exmaple-limit}). The minimal basis of its
  (upward closed) complement in $\setN^d$ is $\{(4,4),(5,3),(6,2)\}$.
  \begin{figure}[htpb]
    \begin{center}
      \def\nd#1#2{\fill (#1,#2) circle (3.5pt);}%
      \def\cir#1#2{\fill (#1,#2) circle (5pt);}%
      \def\cirq#1#2{\draw (#1,#2) circle (9pt);}%
      \def\sq#1#2{\draw (#1-.2,#2-.2) rectangle (#1+.2,#2+.2);}%
      \begin{tikzpicture}[scale=.3,grid/.style ={black!50}]
        \draw[grid] (0,0) grid (10.5,5.5); \draw[->,>=stealth']
        (-.5,0)--(11,0); \draw[->,>=stealth'] (0,-.5)--(0,6);
        \draw[color=black,fill=black,join=round,opacity=.2]
        (10.2,0)--(0,0)--(0,5.2)--(3,5.2)--(3,3)--(4,3)--(4,2)--(5,2)--(5,1)--(10.2,1);
        \cir {13.5} 0\cir {13.5} 1\cir5 2\cir4 3\cir 0 8\cir 1 8\cir 2
        8\cir 3 8 \cirq 3 8 \cirq 4 3\cirq 5 2\cirq {13.5}1;
        \draw(12,0) node {$\cdots$};
        \draw(12,1) node {$\cdots$};
        \draw(12,-1) node {$\cdots$};
        \draw(0,7.2) node {$\vdots$};
        \draw(1,7.2) node {$\vdots$};
        \draw(2,7.2) node {$\vdots$};
        \draw(3,7.2) node {$\vdots$};
        \draw(-1,1) node {\small\smaller 1};
        \draw(-1,2) node {\small\smaller 2};
        \draw(-1,3) node {\small\smaller 3};
        \draw(1,-1) node {\small\smaller 1};
        \draw(2,-1) node {\small\smaller 2};
        \draw(3,-1) node {\small\smaller 3};
        \draw(4,-1) node {\small\smaller 4};
        \draw(5,-1) node {\small\smaller 5};
        \draw(-1,8) node {\small\smaller$\omega$};
        \draw(13.5,-1) node {\small\smaller $\omega$};
      \end{tikzpicture}
    \end{center}  
    \caption{A set $\vec D$ (grayed), elements of a basis (\basiselt and \minbasiselt)
      and of its minimal basis (\minbasiselt)}
    \label{fig:exmaple-limit}
  \end{figure}
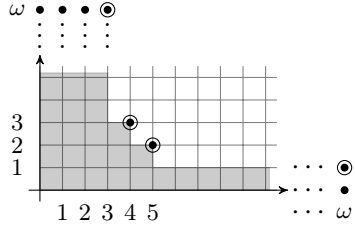
\end{example}

\section{Vector Addition Systems}
\label{sec:vas}

\begin{definition}
  A \emph{Vector Addition System with one zero-test} (shortly \vasz)
  of dimension~$d$ is a tuple $\VV = \tup{A, \az, \delta, \ini}$, where $A$
  is a finite alphabet of \emph{actions}, $\az \not\in A$ is called the
  \emph{zero-test}, $\delta : A \cup \{\az\} \to \setZ^d$ is a mapping, and $\ini \in
  \setN^d$ is the \emph{initial state}.
\end{definition}

Other equivalent formalisms exist, for instance with states, or with
multiple zero-tests transitions that test the same counter for
zero. For now, we stick to the simplest version, and we shall
introduce states in Section \ref{sec:rcsrp}.

\medskip

Intuitively, a \vasz works with $d$ counters, one for each component,
whose initial values are given by \ini. Executing action $a \in A \cup \{
\az \}$ translates the counter values according to $\delta(a) \in \setZ^d$. The
mapping $\delta$ extends to a monoid morphism $\delta: (A \cup \{ \az \})^*\to\setZ^d$, so
that $\delta(\varepsilon)=\vec 0$ and $\delta(uv) = \delta(u) + \delta(v)$ for $u,v \in (A \cup
\{\az\})^*$. More formally, a \vasz $\VV = \tup{A, \az, \delta, \ini}$ of
dimension $d$ induces a transition relation ${\to} \subseteq \setN^d \times A \times \setN^d$
with:
\begin{equation}
  \label{eq:vasz}
  \left\{
    \begin{array}{rcll}
    \vec x \xrightarrow{~a~} \vec y & \text{if} & \delta(a) = \vec y - \vec x
    & \text{~~~for all $a \in A$} \\[1ex]
    \vec x \xrightarrow{\;\az\;} \vec y & \text{if} & 
    \delta(\az) = \vec y-\vec x, \text{ and } \vec x(1) = 0.
  \end{array}
  \right.
\end{equation}
We extend this relation to words by $\vec
x \ru{~\varepsilon~} \vec x$ and $\vec x \ru{~uv~} \vec z$ if there exists $\vec
y$ such that $\vec x \ru{~u~} \vec y \ru{~v~} \vec z$. We say that $u \in (A
\cup \{\az\})^*$ is \emph{fireable} from $\vec x$ if there exists $\vec y$
such that $\vec x \ru{~u~} \vec y$. When there may be ambiguity on the
\vasz, we will write $\xrightarrow{~u~}_\VV$ instead of
$\xrightarrow{~u~}$.

\begin{definition}
  A \emph{Vector Addition System (VAS) of dimension $d$} is a tuple
  $\tup{A, \delta, \ini}$, where $A$ is a finite alphabet, $\delta : A \to 
  \setZ^d$ is a mapping and $\ini\in\setN^d$ is the \emph{initial state}.
\end{definition}

A VAS is a particular \vasz: choosing $\az \not\in A$, this VAS is formally
equivalent to the \vasz $\tup{A, \az, \delta', \ini}$, where $\delta'$
extends $\delta$ by $\delta'(\az) = (-1, 0, ..., 0)$ (\ie, \az can never be fired).

\medskip For a \vasz or a VAS \VV of dimension~$d$, the
\emph{reachability set} $\poststar(\VV)$ and the \emph{cover}
$\cover{}(\VV)$ of $\VV$ are the following subsets of $\setN^d$:
\begin{align*}
  \poststar(\VV) & = \bigl\{ \vec y \in\setN^d\mid \exists u \in (A \cup \{\az\})^*, 
  \ini \xrightarrow{~u~} \vec y \bigr\}, \\
  \cover{}(\VV) & = \down{}\poststar(\VV).
\end{align*}

We call elements of $\poststar(\VV)$ \emph{reachable states} (also
called reachable markings in related work). The reachability (resp.\
coverability) problem consists in deciding membership in
$\poststar(\VV)$ (resp.\ in $\cover{}(\VV)$). Reachability is
decidable for VAS \cite{Mayr:81,Kosaraju:82,LEROUX-POPL2011} and
\vasz~\cite{Reinhardt:08,MFCS:11}.

\begin{theorem}
  \label{thm:reach}
  Given a VAS or \vasz $\VV$, the reachability problem for $\VV$ is
  decidable.
\end{theorem}

Testing membership in the cover set is much easier, and one even gets
a more precise
result~\cite{Karp&Miller:69,Finkel&Schnoebelen:01,Finkel&Goubault-Larrecq:09:b}:

\begin{theorem}
  \label{thm:cover}
  Given a VAS \VV, one can effectively compute a (finite) basis of 
  $\cover{}(\VV)$. 
\end{theorem}

Observe that given a (finite) basis $\vec B$ of a downward closed set
$\vec D\subseteq\setN^d$, one can effectively test membership in~$\vec D$, since
$\vec D=\setN^d\cap\dc \vec B$ by \eqref{eq:2} and~\eqref{eq:3}. Therefore,
Theorem~\ref{thm:cover} implies that one can effectively decide
membership in $\cover{}(\VV)$.

Computing a finite basis of the cover makes it also possible to decide
whether two VAS have the same cover, since from a finite basis, one
can also compute the minimal basis, which is canonical. Likewise, one
can decide inclusion of covers. Finally, Theorem~\ref{thm:cover}
implies that one can decide \emph{place-boundedness}, that is, whether
the projection of $\poststar(\VV)$ on some given component is
bounded. In the next three sections, we shall show that one can also
effectively compute a finite basis for the cover of a $\vasz$.

\section{Limits of reachable states of a VAS}
\label{sec:set-limits-reachable}

As observed above, for $\vec M \subseteq \setN^d$, one can immediately construct
an algorithm deciding membership in $\vec M$ from an algorithm
deciding membership in $\Lim \vec M$, since $\vec M = \setN^d \cap \Lim \vec
M$ by~\eqref{eq:2}. However, the converse is not true. Let us explain
two reasons for this.

\begin{enumerate}[leftmargin=*,itemsep=.5ex,label={\!\!$\alph*.$\!},ref=$\alph*$]
\item First, even if $\vec M$ is recursive, it may happen that $\Lim
  \vec M$ is not. We recall here an example from
  \cite[Prop.~2.4]{Finkel&McKenzie&Picaronny:well-structured-framework-analysing-petri:2004:b}.
  Let $T_0,T_1,\ldots$ be an effective enumeration of Turing machines. Let
  $\alpha(k,\ell)=\bigl|\{j\leq k\mid T_j\text{ halts in at most $\ell$ steps on
    $\varepsilon$}\}\bigr|$ and $\vec M=\bigl\{(k,\ell, \alpha(k,\ell))\mid k,\ell\geq0\bigr\}$. It is
  easy to describe an algorithm computing $\alpha(k,\ell)$ given $k,\ell\in\setN$, and
  therefore also an algorithm to decide membership in~$\vec
  M$. However, $\Lim\vec M$ is not recursive, since the halting
  problem reduces to it. Indeed, $(k,\omega,m)\in\Lim\vec M$ means that
  exactly $m$ machines among $T_0,\ldots,T_k$ halt on the empty
  word. Therefore, $T_k$ halts on $\varepsilon$ if and only if there exists $m\leq
  k+1$ such that $(k,\omega,m)\in\Lim\vec M$ and $(k-1,\omega,m-1)\in\Lim\vec M$.

\item\label{item:1} Second, even if $\Lim \vec M$ is recursive, one
  may not be able to effectively derive an algorithm deciding
  membership in $\Lim\vec M$ from a description of $\vec M$ (such as a
  data structure, or an algorithm deciding membership in $\vec M$). As
  an example, consider the reachability~set~$\vec M$ of a lossy
  counter machine (see again~\cite{Mayr:03}, or
  \cite{Schnoebelen:Lossy-counter-machines-decidability:2010:a} for a
  survey). An algorithm to decide membership of $\vec x$ in $\vec M$
  is to compute the bases of the upward closed sets
  $\text{Pre}^i(\uc\vec x)$ for $i=0,1,2,... $, where $\text{Pre}(X)$
  denotes the set of predecessors of $X$. The sequence stabilizes,
  since it consists only of upward closed sets.  Moreover, due to the
  lossy behavior, $\vec M$ is downward closed. Therefore, it admits a
  finite basis $\vec B$, so that $\Lim\vec M=\dc\vec B$ is
  recursive. However, there is no algorithm taking as input a lossy
  counter machine and a vector $\vec x\in\setN^d_\omega$, and deciding membership
  of $\vec x$ in $\Lim\vec M$, where $\vec M$ is the reachability
  set. Indeed, the set $\vec M$ is infinite if and only if $\Lim \vec
  M$ contains some vector of $\setN^d_\omega$ having at least an
  $\omega$-component. Therefore, the existence of such an algorithm would
  imply that the boundedness problem (\ie, whether the reachability
  set is finite) is co-recursively enumerable, which is not the case:
  boundedness for lossy counter machines is $\Sigma_1^0$-complete.
\end{enumerate}

\medskip\noindent The main result of this section considers the case
where $\vec M$ is the reachability set of a~\vas~\VV.  Since
$\cover{}(\VV)=\dc\poststar(\VV)=\setN^d\cap \Lim\dc\poststar(\VV)=\setN^d\cap
\dc\Lim\poststar(\VV)$ (where the last two equalities follow
from~\eqref{eq:2} and \eqref{eq:5}), one can by
Theorem~\ref{thm:cover} effectively compute a basis of
$\dc\Lim\poststar(\VV)$. However, since $\Lim\poststar(\VV)$ is not
necessarily downward closed, this does not directly entail an
algorithm for deciding membership in this set.

\begin{theorem}\label{thm:reclim}
  Given a \vas \VV and $\vec x\in\setN^d_\omega$, one can decide whether $\vec x\in\Lim \poststar(\VV)$.
\end{theorem}
We establish Theorem~\ref{thm:reclim} by describing two
semi-algorithms proving that $\Lim\poststar(\VV)$ and its complement
in $\setN_\omega^d$ are both recursively enumerable sets.
%
%
Let us start with the most interesting direction. We shall prove that
$\Lim\poststar(\VV)$ is recursively enumerable, by introducing
\emph{productive sequences}, a notion inspired by
Hauschildt~\cite{Hauschildt:90}.
\begin{definition}
  \label{def:productive}
  Let $\VV = \tup{A,\delta,\ini}$ be a VAS, and let $\pi=(u_i)_{0\leq i\leq k}$ be
  a sequence of words over $A$. We say that $\pi$ is \emph{productive}
  in \VV for a word $v=a_1\cdots a_k\in A^*$ if the words
  \begin{equation*}
    u_0^na^{}_1u_1^n\cdots a^{}_ku_k^n,\qquad n\geq1
  \end{equation*}
  are all fireable from \ini. 
\end{definition}

In particular, if $\pi$ is productive for $v$, the state $\ini + \delta(v)+
n\delta(\pi)$ is a reachable state in \VV, where
$\delta(\pi)=\sum_{i=0}^k\delta(u_i)$. Definition~\ref{def:productive} shows that the
set $\bigl\{(\pi,v)\mid\pi\text{ productive in \VV for }v\bigr\}$ is
co-recursively enumerable. The following characterization immediately
gives an algorithm to decide membership in this set, showing that it
is actually recursive.

\begin{lemma}
  \label{lem:caraprod}
  A sequence $\pi=(u_i)_{0\leq i\leq k}$ is \emph{productive} in \VV for a
  word $a_1\cdots a_k$ if and only~if 
  \begin{enumerate}[label=\quad$(\arabic*)$,ref=$(\arabic*)$]
  \item\label{item:2} the partial sums $\delta(u_0)+\cdots+\delta(u_j)$ are nonnegative for 
    every $j\in\{0,\ldots,k\}$, and
  \item\label{item:3} the word $u_0 a_1 u_1 \cdots a_k u_k$ is fireable from \ini.
  \end{enumerate}
\end{lemma}

\begin{proof}
  Let us introduce the states $\vec{y}_0=\ini$ and
  $\vec{y}_j=\ini+\delta(a_1\cdots a_j)$ for $j\in\{1,\ldots,k\}$, and the partial sums
  $\vec x_{-1}=\vec 0$ and $\vec{x}_j=\delta(u_0)+\cdots+ \delta(u_j)$ for
  $j\in\{0,\ldots,k\}$. We put $u[-1,n]=\varepsilon$, $u[j,n]=u_0^na^{}_1u_1^n\cdots
  a^{}_ju_j^n$ for $j\geq0$, and $v[j,n]=u[j,n]a_{j+1}$.

  \smallskip If $\pi$ is productive for $v=a_1\cdots a_k$, then $u[k,n]$ is
  fireable from $\ini$ for all $n\geq1$. Therefore, $u[j,n]$ is also
  fireable from $\ini$ for $j\leq k$. We deduce that
  $\ini+\delta(u[j,n])=\vec{y}_j+n\vec{x}_j$ is nonnegative for every
  $n\in\setN$. In particular $\vec{x}_j\geq \vec0$.  We have
  proved~\ref{item:2}, and~\ref{item:3} is obvious.

  Conversely, assume that~\ref{item:2} and \ref{item:3} both hold. For
  all $n\geq1$, we have to show that $u[k,n]$ is fireable from $\ini$,
  \ie, that $\ini+\delta(w)\geq\vec0$ for any nonempty prefix $w$ of
  $u[k,n]$. Such a prefix is of the form $v[j-1,n]u_j^pu'_j$ for some
  $0\leq j\leq k$, $0\leq p<n$, and some prefix $u'_j$ of $u^{}_j$. By
  rearranging terms, we obtain
  \begin{align*}
    \ini+ \delta(v[j-1,n]u_j^pu'_j) &=\ini+\delta(u_0^na^{}_1u_1^n\cdots a^{}_{j-1}u_{j-1}^na_ju_j^pu'_j)\\
    &=\ini + \delta(u_0a_1u_1\cdots a_ju'_j)+ (n-1)\vec{x}_{j-1}+\delta(u_j^p)\\
    &=\ini + \delta(u_0a_1u_1\cdots a_ju'_j)+ (n-p-1)\vec{x}_{j-1}+p\vec{x}_j.
  \end{align*}
  By~\ref{item:2}, we have $\vec{x}_{j-1},\vec{x}_j\geq\vec0$.
  By~\ref{item:3}, the word $u_0a_1u_1\cdots a_ku_k$ is fireable
  from~$\ini$, and in particular, $\ini+\delta(u_0a_1u_1\cdots a_ju'_j)\geq \vec
  0$. Therefore, $\ini+\delta(v[n,j-1]u_j^pu'_j)\geq\vec0$, which proves that
  $u_0^na^{}_1u_1^n\cdots a^{}_ku_k^n$ is fireable from $\ini$. We have
  shown that $\pi$ is productive for~$a_1\cdots a_k$.
\end{proof}

We will now show in Proposition~\ref{prop:re} below that limits of
reachable states are witnessed by productive sequences. Its essential
argument is Higman's Lemma. We recall that an ordering $\preceq$ is
\emph{well} if every infinite sequence $(\ell_n)_{n\in\setN}$ admits an
infinite increasing subsequence $(\ell_{n_k})_{k\in\setN}$: $\ell_{n_0}\preceq \ell_{n_1}\preceq
\ell_{n_2}\preceq\cdots$.\label{wqo} The pointwise ordering over $\setN^d$ or over
$\setN^d_\omega$ is well (Dickson's~Lemma).

\paragraph{\bfseries Higman's Lemma.} Let $\Sigma$ be a (possibly infinite)
set.  Given an ordering $\preceq$ over $\Sigma$, let $\embed{\preceq}$ be the ordering
over $\Sigma^*$ defined as follows: for $u,v\in \Sigma^*$, we have $u\embed\preceq v$ if
$u=a_1\cdots a_n$ with $a_i\in \Sigma$, $v=v_0b_1v_1\cdots v_{n-1}b_nv_n$, with $v_i\in
\Sigma^*$, $b_j\in \Sigma$, and for all $i=1,\ldots,n$, we have $a_i\preceq b_i$. In other
words, $u$ is obtained from $v$ by removing some letters, and then
replacing some of the remaining letters by smaller ones. Higman's
Lemma is the following result. See for instance~\cite{DIE05B} for a
proof.

\begin{lemma}[Higman]
  \label{thm:higman}
  If $\preceq$ is a well ordering over $\Sigma$, then $\embed\preceq$ is a well
  ordering over~$\Sigma^*$.
\end{lemma}

We extend the multiplication over $\setN_\omega$ by $\omega\cdot0=0=0\cdot\omega$ and $\omega\cdot
k=\omega=k\cdot\omega$ if $k\neq0$. This multiplication then extends componentwise to
the scalar multiplication of $\setN^d_\omega$ by $\setN_\omega$.

\begin{proposition}
  \label{prop:re}
  Let $\VV=\tup{A,\delta,\ini}$ be a VAS. Then
  \begin{equation*}
    \Lim\poststar(\VV)=\bigl\{\ini+\delta(v)+\omega\delta(\pi)\mid 
    v \in A^* \text{and $\pi$ productive in \VV for $v$}\bigr\}.
  \end{equation*}
\end{proposition}

\begin{proof}
  For the inclusion from right to left, if $\pi$ is a productive
  sequence for a word $v$, then $\ini+\delta(v)+\omega\delta(\pi)$ is the limit of the
  sequence $(\vec{x}_n)_{n\in\setN}$ with $\vec{x}_n=\ini+\delta(v)+n\delta(\pi)$, which
  is a reachable state by Definition~\ref{def:productive}. We prove
  the reverse inclusion thanks to Higman's~lemma. We follow the
  approach of Jan\v{c}ar introduced in~\cite[Section~6]{Petr199071}.
  
  \smallskip
  Let us first introduce a well ordering $\sqsubseteq$ over $\poststar(\VV)$,
  using a temporary ordering~$\preceq$. Consider the infinite set 
  $\Sigma=A\times\setN^d_\omega$. This set is well ordered by $\preceq$, defined 
  by:
  \begin{equation*}
    (a,\vec{y})\preceq(b,\vec{z}) \text{ if and only if } 
    a=b \text{ and } \vec{y}\leq \vec{z}.
  \end{equation*}
  Since $\preceq$ is a well ordering, Higman's lemma shows that $\embed\preceq$ is
  a well ordering over $\Sigma^*$. We associate to every reachable state
  $\vec{y}\in \poststar(\VV)$ a word $\alpha_{\vec{y}}$ in $\Sigma^*$ as follows:
  since $\vec{y}$ is reachable, the set $V_{\vec y}=\{v\in
  A^*\mid\ini\xrightarrow{v}\vec y\}$ is nonempty. Let us choose
  arbitrarily some $v_{\vec y}$ in $V_{\vec y}$ (the actual choice is
  irrelevant, one can choose for instance the minimal element of
  $V_{\vec y}$ wrt.\ the lexicographic ordering). Let $v_{\vec y} =
  a_1\cdots a_k$, with $k\geq0$ and $a_i \in A$. We introduce the sequence
  $(\vec{y}_i)_{0\leq i\leq k}$ of states defined by $\vec y_0=\ini$, and
  $\vec{y}_i=\ini+\delta(a_1\cdots a_i)$ for $i\geq1$. We let
  \begin{equation*}
    \alpha_{\vec{y}}=(a_1,\vec{y}_1)\cdots(a_k,\vec{y}_k).
  \end{equation*}
  We define the ordering $\sqsubseteq$ over $\poststar(\VV)$ by 
  $\vec{y} \sqsubseteq \vec{z}$ if $\alpha_{\vec{y}} \embed\preceq 
  \alpha_{\vec{z}}$ and $\vec{y}\leq \vec{z}$. Since the orderings 
  $\embed\preceq$ over $\Sigma^*$ and $\leq$ over $\setN^d$ are well, we deduce 
  that $\sqsubseteq$ is a well ordering over~$\poststar(\VV)$.

  \smallskip Now, let us pick $\vec{x} \in
  \Lim\poststar({\VV})$: $\vec x$ is the limit of a sequence
  $(\vec{x}_k)_{k\in\setN}$ of reachable states. By extracting a
  subsequence if necessary, one can assume that for every index~$i$:
  \begin{enumerate}[label=$(\roman*)$]
  \item\label{item:4} if $\vec{x}(i)<\omega$, then $\vec{x}_k(i)$ is
    constant, equal to $\vec{x}(i)$, and
  \item\label{item:5} if $\vec{x}(i) = \omega$, then $(\vec{x}_k(i))_{k\in\setN}$
    is strictly increasing.
  \end{enumerate}
  Denote by $\alpha_j$ the word $\alpha_{\vec{x}_j}$ associated to the reachable
  state $\vec{x}_j$. Since $\sqsubseteq$ is a well ordering, there exist $m<n$
  such that $\vec{x}_m \sqsubseteq \vec{x}_n$. By construction of $\alpha_m$, there
  exists a word $v = a_1\cdots a_k$ with $a_j \in A$ such that the sequence
  $(\vec{y}_j)_{1\leq j\leq k}$ defined by $\vec{y}_j = \ini+\delta(a_1\cdots a_j)$
  for every $j \in \{1,\ldots,k\}$ satisfies:
  \begin{equation*}
    \alpha_m = (a_1,\vec{y_1}) \cdots (a_k,\vec{y}_{k})
  \end{equation*}
  Since $\alpha_m\embed\preceq\alpha_n$ and by definition of $\embed\preceq$, there exist a
  sequence $(\vec{z}_j)_{1\leq j\leq k}$ of states with $\vec{y}_j \leq
  \vec{z}_j$, and a sequence $(\beta_j)_{0\leq j\leq k}$ of words in $\Sigma^*$ such
  that the following equality holds:
  $$
  \alpha_n = \beta_0 (a_1,\vec{z}_1) \beta_1 \cdots (a_k,\vec{z}_k) \beta_k
  $$
  We call \emph{label} of a word $(b_1,\vec{t}_1)\cdots(b_\ell,\vec{t}_\ell)$
  over $\Sigma$ the word $b_1\cdots b_\ell$ over~$A$. Consider the sequence
  $\pi=(u_j)_{0\leq j\leq k}$ where $u_j$ is the label of $\beta_j$. Since
  $\vec{x}_m$ and $\vec{x}_n$ are reachable, we have by definition of
  $\alpha_m$ and $\alpha_n$:
  \begin{align}
    \begin{split}
      \label{eq:6}
      \ini\xrightarrow{a_1} \vec{y}_1 \cdots& \xrightarrow{a_k} \vec{y}_k=\vec x_m
      \\
      \ini\xrightarrow{u_0a_1} \vec{z}_1 \cdots& \xrightarrow{u_{k-1}a_k}
      \vec{z}_k \xrightarrow{u_k}\vec{x}_n
    \end{split}
  \end{align}
  From \eqref{eq:6}, we obtain in particular
  \begin{equation}
    \label{eq:7}
    \vec{z}_j=\vec{y}_j+\delta(u_0)+\cdots+\delta(u_{j-1})\text{ for every $j\in\{1,\ldots,k\}$}
  \end{equation}
  and in the same way,
  \begin{equation}
    \label{eq:8}
    \vec{x}_n=\vec{x}_m+\delta(\pi)
  \end{equation}
  Using \eqref{eq:7} with $\vec{y}_j\leq \vec{z}_j$ for $1\leq j\leq k$, and
  \eqref{eq:8} with $\vec{x}_m\leq \vec{x}_n$, we deduce that $\pi$
  satisfies property~\ref{item:2} of Lemma~\ref{lem:caraprod}. Since,
  by~\eqref{eq:6}, it also satisfies~\ref{item:3}, it is productive
  for $v$.

  \smallskip\noindent It remains to prove that $\vec{x}=\vec{y}$ where
  $\vec{y}=\ini+\delta(v)+ \omega\delta(\pi)$. Let $i\in \{1,\ldots,d\}$.
  \begin{itemize}
  \item If $\vec{x}(i)<\omega$ then by~\ref{item:4}, we get $\vec{x}_m(i)=
    \vec{x}(i)=\vec{x}_n(i)$, so using~\eqref{eq:8}, we obtain
    $\delta(\pi)(i)=0$. Since we have $\vec{x}_n=\ini+\delta(v)+\delta(\pi)$
    by~\eqref{eq:6}, we deduce that
    $\vec{x}(i)=\vec{x}_n(i)=\ini(i)+\delta(v)(i)=\vec{y}(i)$.

  \item If $\vec{x}(i)=\omega$, then by~\ref{item:5} $\vec{x}_m(i)<\vec{x}_n(i)$. We
    deduce from~\eqref{eq:8} that $\delta(\pi)(i)>0$. Therefore,
    $\vec{x}(i)=\omega=\vec{y}(i)$. 
  \end{itemize}
  Finally, $\vec{x}=\vec{y}$, and we have proved that there exists a
  productive sequence $\pi$ for a word $v$ such that
  $\vec{x}=\ini+\delta(v)+\omega\delta(\pi)$.
\end{proof}

Proposition~\ref{prop:re} and Lemma~\ref{lem:caraprod} provide a
semi-algorithm to test whether a given vector $\vec x\in \setN^d_\omega$ belongs
to $\Lim\poststar(\VV)$: it suffices to enumerate the pairs $(\pi,v)$,
where $\pi$ is productive for $v$, and to check whether $\vec x =
\ini+\delta(v)+\omega\delta(\pi)$.

\bigskip It is easier to prove that the \emph{complement} of
$\Lim\poststar(\VV)$ is recursively enumerable.  Consider $\vec
y\in\setN^d_\omega$. We introduce $d$ distinct additional elements $b_1,\ldots,b_d
\not\in A$. Let $B=\{b_1,\ldots,b_d\}$. We now introduce the \vas $\VV_{\vec y}
= \tup{A\uplus B,\delta_{\vec{y}},\ini}$, where $\delta_{\vec{y}}$ extends $\delta$ by:
\begin{equation*}
  \delta_{\vec{y}}(b_i)=\begin{cases}
    \vec0 & \text{ if $\vec y(i)<\omega$,}\\
    -\vec{e}_i & \text{ if $\vec y(i)=\omega$.}
  \end{cases}
\end{equation*}
Finally, we define from $\vec y$ a sequence $(\vec{y}_\ell)_\ell$ converging to $\vec y$, by
\begin{math}
  \vec{y}_\ell(i)=
    \begin{cases}
      \vec y(i)& \text{if $\vec y(i)<\omega$,}\\
      \ell&\text{if $\vec y(i)=\omega$.}
    \end{cases}
\end{math}

\begin{lemma}
  \label{lem:co-re}
  Let $\VV_{\vec y}$ and $(\vec{y}_\ell)_\ell$ constructed from $\vec y$ as 
  above. Then,
  \begin{equation}
    \label{eq:9}
    \vec y\not\in\Lim\poststar(\VV) \Longleftrightarrow \exists \ell\in\setN,\ 
    \vec{y}_\ell\notin\poststar(\VV_{\vec y}).
  \end{equation}
  In particular, the complement of $\Lim\poststar(\VV)$ is effectively recursively 
  enumerable.
\end{lemma}

\begin{proof}
  We prove the following, which is equivalent to \eqref{eq:9}:
  $$
  \vec y\in\Lim\poststar(\VV) \Longleftrightarrow 
  \forall \ell\in\setN,~   \vec{y}_\ell\in\poststar(\VV_{\vec y}).
  $$
  Assume that $\vec y\in\Lim\poststar(\VV)$. Fix $\ell\in\setN$. There exists a 
  sequence $(\vec{z}_n)_n$ of  elements of $\poststar(\VV)$ such that 
  $\lim_n\vec{z}_n=\vec y$, so for $n$ large enough, we have for all 
  $i=1,\ldots,d$:
  \begin{itemize}[leftmargin=3em]
  \item $\vec{z}_n(i)=\vec y(i)=\vec{y}_\ell(i)$ if $\vec y(i)<\omega$,
  \item $\vec{z}_n(i)\geq \ell=\vec{y}_\ell(i)$ if $\vec y(i)=\omega$.
  \end{itemize}
  Then $\vec{z}_n\xrightarrow{u}\vec{y}_\ell$ in $\VV_{\vec y}$, 
  with $u=\prod_{i=1}^d{b_i}^{\vec{z}_n(i)-\vec{y}_\ell(i)}$. Since
  $\vec{z}_n$ is reachable from $\ini$ (already in $\VV$), we deduce
  that $\vec{y}_\ell\in\poststar(\VV_{\vec y})$.

  \smallskip Conversely, assume that $\vec{y}_\ell\in\poststar(\VV_{\vec
    y})$ for all $\ell$, and let $u_\ell\in(A\cup B)^*$ such that
  $\ini\xrightarrow{u_\ell}\vec{y}_\ell$, in $\VV_{\vec y}$. Consider the
  word $v_\ell$ obtained from $u_\ell$ by erasing all letters of $B$. Since
  $\delta(b)\leq\vec 0$ for $b\in B$, the word $v_\ell$ is still fireable from
  $\ini$, so that $\vec{z}_\ell=\ini+\delta(v_\ell)\in\poststar(\VV)$. Moreover, by
  definition of $\VV_{\vec y}$, $\vec{z}_\ell(i)=\vec{y}_\ell(i)$ if $\vec
  y(i)<\omega$ and $\vec{y}_\ell(i)\leq\vec{z}_\ell(i)$ otherwise. Therefore,
  $\lim_\ell\vec{z}_\ell=\lim_\ell\vec{y}_\ell=\vec y$, and it follows that $\vec
  y\in\Lim\poststar(\VV)$.
  
  \smallskip This shows \eqref{eq:9}. Hence, we can enumerate vectors
  $\vec{y}_\ell$ and test, for each $\vec{y}_\ell$, its membership in
  $\poststar(\VV_{\vec y})$. This proves that $\Lim\poststar(\VV)$ is
  co-recursively enumerable.
\end{proof}

Theorem~\ref{thm:reclim} now follows from Proposition~\ref{prop:re}
and Lemma~\ref{lem:co-re}.

\section{Refined and filtered covers}
\label{sec:filtered-covers}

In this section, we introduce two new notions of covers: refined and
filtered covers. Both are parameterized, and the following inclusions
will hold, regardless of the parameters:
\begin{displaymath}
  \poststar(\VV)\subseteq\mathsf{RefinedCover}(\VV)\subseteq\cover{}({\VV}),\text{\qquad and\qquad}  \mathsf{FilteredCover}(\VV)\subseteq\cover{}(\VV)
\end{displaymath}

\smallskip Let us first introduce the refined cover, a set hybrid
between the reachability and cover sets, that to our knowledge has not
yet been considered. Instead of the downward closure $\cover{}(\VV)$
of $\poststar(\VV)$ wrt.\ the pointwise ordering~$\leq$, we consider
$$\cover{\leq_P}(\VV)=\Pdown{P}{\poststar(\VV)},$$ that is, we
replace $\leq$ with an ordering ${\leq_P}$ over $\setN_\omega^d$ parameterized by a
set of ``positions'' $P\subseteq\{1,\ldots,d\}$:
\begin{equation*}
  \vec x\leq_P\vec y \text {\quad if\quad}
  \begin{cases}
    \vec x(i) =\vec y(i)&\text{ for $i\in P$,} \\
    \vec x(i) \leq\vec y(i)&\text{ for $i\notin P$.}
  \end{cases}
\end{equation*}

The set $P$ contains the components for which we insist on keeping
equality. Thus, $\leq_\emptyset$ is the usual pointwise ordering~$\leq$, while
$\leq_{\{1,\ldots,d\}}$ boils down to equality. Notice that $\leq_P$ is \emph{not}
a well ordering, except if $P=\emptyset$ (\eg, $\setN$ ordered by $\leq_{\{1\}}$
consists only of incomparable elements, since in this case, $\leq_{\{1\}}$
is just equality).

\medskip The ordering $\leq_{\{1\}}$ will be abbreviated as $\lequn$. It is
a natural order to study for a $\vasz$ (recall that the zero-test
occurs on the first component). Indeed, the transition relation of a
\vasz is monotonic with respect to this order: if $\vec{x} \ru{u}
\vec{x}'$ and $\vec x \lequn \vec y$, then there exists $\vec{y}'$
with $\vec {y} \ru{u} \vec {y}'$ and $\vec {x}' \lequn \vec{y}'$. In
words, from a \lequn-larger state than $\vec x$, one can perform the
same transitions as from $\vec x$, and reach a state \lequn-above that
the one reached from $\vec x$. This is clearly not the case if one
uses the pointwise ordering $\leq$ instead of \lequn: some zero-tests may
fail from the largest state and succeed from the smallest one.

More precisely, testing if $\xcover{\lequn}(\VV)$ contains a vector
whose first component is 0 is what we need to design our algorithm
computing the cover of a VAS with one zero test. Unfortunately, the
set $\xcover{\lequn}(\VV)$ cannot be represented by a finite set of
$\lequn$-maximal elements, since it may well have infinitely many of them.
Actually, the following theorem shows that we cannot find a sensible
way to compute a representation of this set, as any representation
would not allow to test for~equality.

\begin{theorem}\label{thm:indecidable}
  Given two VAS $\VV_1$, $\VV_2$, it is undecidable whether
  $\xcover{\lequn}(\VV_1) = \xcover{\lequn}(\VV_2)$.
\end{theorem}

\begin{proof}
  We reduce the equality problem $\poststar(\VV_1)=\poststar(\VV_2)$,
  which is known to be undecidable \cite{Baker73,Hack:76}, to the
  problem of the statement. Let us first consider a VAS
  $\VV=\tup{A,\delta,\ini}$ of dimension $d$. We introduce a VAS
  $\VV'=\tup{A,\delta',\iniprime}$ of dimension $d+1$ that counts in the
  first component the sum of the other components. Formally,
  $\iniprime=\bigl(\sum_{i=1}^d\ini(i),\ini\bigr)$ and
  $\delta'(a)=\bigl(\sum_{i=1}^d\delta(a)(i),\delta(a)\bigr)$ for every $a\in A$.  Observe
  that the following equivalence holds:
  \begin{equation*}
    (n,\vec{x}) \in \poststar(\VV')~\Longleftrightarrow~\vec{x}\in
    \poststar(\VV)~\text{ and }n=\sum_{i=1}^d\vec{x}(i).
  \end{equation*}
  Finally, consider two VAS $\VV_1$ and $\VV_2$, and just observe that
  $\poststar(\VV_1)=\poststar(\VV_2)$ if and only if
  $\xcover{\lequn}(\VV_1') = \xcover{\lequn}(\VV_2')$.
\end{proof}

So, we cannot hope for a useful representation of the sets
$\xcover{\leq_P}(\VV)$. However, one can capture the needed information
differently, by replacing the downward closure $\Pdown{P}$ in
$\cover{\leq_P}(\VV)=\Pdown{P}\poststar(\VV)$ with another operator
$\xdown f{}$, parameterized by a vector $\vec f$ of $\setN^d_\omega$ (the
letter $\vec f$ stands for \emph{filter}). Informally, $\xdown f{}\vec
M$ is a downward closure taking into account only elements of $\vec M$
that agree with $\vec f$ on its finite components. Other elements will
just be discarded. Formally, for $\vec f\in\setN^d_\omega$ and $\vec M\subseteq\setN_\omega^d$,~we
define the \emph{filtered cover} $\xdown f{\vec M}$ by:
\begin{align*}
  \filter{\vec M}f&=\Bigl\{\vec x\in \vec M\mid\bigwedge_{i=1}^d \bigl[\,\vec f(i)<\omega\Longrightarrow\vec x(i)=\vec f(i)\bigr]\Bigr\},\\
  \xdown f{\vec M}&=\big\downarrow\filter{\vec M}f.
\end{align*}

Observe that $\xdown f{\vec M}$ is a downward closed subset of
$\down{} \vec M$, and that $\xdown {{}(\omega,\omega,\ldots,\omega)}{\vec M}=\dc{}{\vec
  M}$.
Elements of the minimal basis of ${\xdown f{\vec M}}$ agree with $\vec
f$ on components $i$ where $\vec f(i)<\omega$.  One can check that the
limit and filter operators commute:
$$\filter {\Lim\vec M}{f}=\Lim\filter{\vec M}{f}.$$
Since the limit and the downward closure operators also commute
(see~\eqref{eq:5}), we obtain
\begin{equation}
  \label{eq:10}
  \xdown {f}{\Lim\vec M}=\Lim\xdown{f}{\vec M}.
\end{equation}

The motivation for considering filtered covers is that, for $\vec
f=(0,\omega,\ldots,\omega)\in\setN_\omega^d$ and $\vec M=\poststar{}(\VV)$ where $\VV$ is a \vas
of dimension $d$, the set $\xdown f{\vec M}$ captures all information
we need to overcome the difficulty described on
page~\pageref{sec:difficulty}. Moreover, contrary to the refined cover
of a \vas, all its filtered covers are computable, as stated in
Theorem~\ref{thm:RE} below. Our goal in this section is to describe an
algorithm computing a filtered cover of a \vas. Our algorithm both
refines Karp and Miller's one to compute the usual cover, and
generalizes Theorem~\ref{thm:reclim}.

\begin{theorem}
  \label{thm:RE}
  Let $\VV$ be a \vas.  Given $\vec f\in\setN^d_\omega$, one can compute a basis
  of $\xdown f\poststar(\VV)$.
\end{theorem}

Karp and Miller's algorithm computing $\cover{}(\VV)$ corresponds to
the case $\vec f=(\omega,\ldots,\omega)$.  Since $\vec M$ and $\Lim\vec M$ have the
same bases (by definition~\eqref{eq:3} of a basis), computing a basis
of $\xdown f\poststar(\VV)$ is the same as computing a basis of
$\Lim\xdown f\poststar(\VV)$, \ie, by~\eqref{eq:10}, of $\xdown
f\Lim\poststar(\VV)$. We first reduce this computation to a decision
problem, as in~\cite[Th.~2.10]{Valk&Jantzen:85a}.

\smallskip
 For $\vec M\subseteq\setN_\omega^d$, let us introduce the following set:
\begin{align*}
  \fyset{M}&=\bigl\{(\vec{f},\vec{y})\in\setN_\omega^d\times\setN_\omega^d \mid \vec{y}\in \xdown f{\vec M}\bigr\}.
\end{align*}

\begin{lemma}
  \label{lem:vj-down}
  Let $\vec M\subseteq\setN_\omega^d$ be a limit closed set. From an algorithm solving
  the membership problem for $\fyset{M}$, one can construct an
  algorithm which, given an input vector $\vec f\in\setN_\omega^d$, outputs a
  basis of $\xdown f{\vec M}$.
\end{lemma}

\begin{proof}  
  Observe that $\xdown f{\vec M}$ has the same bases as $\setN^d\cap\xdown
  f{\vec M}$. Now, if ${\vec D}\subseteq\setN^d$ is downward closed, one can
  compute a basis $\vec B_{\vec D}\subseteq\setN^d_\omega$ of ${\vec D}$ from a basis
  $\vec B_{\vec U}$ of its (upward closed) complement ${\vec U}=\setN^d\setminus
  {\vec D}$: an algorithm generates all candidates $\vec B_{\vec D}$
  for a basis of $\vec D$, (\ie, all finite subsets of the countable
  set $\setN^d_\omega$), and checks for each candidate whether it is indeed a
  basis of ${\vec D}$, \ie, that the union of the sets $\setN^d\cap\dc {\vec
    B_{\vec D}}$ and $\setN^d\cap\uc {\vec B_{\vec U}}$ is $\setN^d$, and that
  their intersection is empty. This property is Presburger definable,
  whence decidable.

  Hence, to compute a basis of ${\xdown f\vec M}$ given $\vec f$, it
  suffices to compute a basis of ${\setN^d\setminus\xdown f\vec M}$.
  Now,~\cite[Th.~2.10]{Valk&Jantzen:85a} describes an algorithm
  computing such a basis from an algorithm deciding, given $\vec
  y\in\setN_\omega^d$, whether $(\setN^d\setminus\xdown f\vec M)\cap\dc\vec y=\emptyset$, or
  equivalently, whether $\setN^d\cap\dc\vec y\subseteq\xdown{f}{\vec M}$.  Note that
  $\vec y$ may have some components whose value is~$\omega$, so, if $\vec
  M$ were an arbitrary set, it might happen that $\dc\vec
  y\not\subseteq\xdown{f}{\vec M}$. However, $\vec M$ is limit closed, and
  therefore $\setN^d\cap\dc\vec y\subseteq\xdown{f}{\vec M}$ is equivalent to $\vec
  y\in\xdown{f}{\vec M}$, that is, to $(\vec f,\vec y)\in\fyset{M}$.
\end{proof}

Notice that Lemma~\ref{lem:vj-down} requires as input an algorithm
solving the membership problem in $\fyset{M}$, \ie, a unique algorithm
solving the membership of $\vec{y}$ in $\xdown f{\vec M}$ where
$\vec{f}$ is an input parameter. This hypothesis cannot be weakened by
just assuming that for each $\vec{f}$ we have an algorithm deciding
the membership of $\vec{y}$ in $\xdown f{\vec M}$. In fact this
hypothesis is a tautology, since the set $\xdown f{\vec M}$ is
recursive, as every downward closed set. The lemma becomes clearly
wrong without any condition on $\vec{M}$.

\smallskip We will now reduce membership in
$\fyset{{}\Lim\poststar(\VV)}$ to a similar problem involving refined
covers. The next lemma provides a relationship between the sets
$\xdown f{\vec M}$ and $\Pdown{P}{\vec M}$.

\begin{lemma}
  \label{lem:fp}
  Let $\vec{M}\subseteq\setN_\omega^d$, $P\subseteq\{1,\ldots,d\}$, and $\vec{y}\in\setN_\omega^d$. Define $\vec{f}\in\setN_\omega^d$ by
  \begin{equation}
    \label{eq:11}
    \vec{f}(i)=
    \begin{cases}
      \vec{y}(i)&\text{ if }i\in P\text{, and}\\
      \omega&\text{ otherwise.}
    \end{cases}
  \end{equation}
  Then we have:
  \begin{equation}
    \label{eq:a:1}
    \vec y\in \Pdown{P}{\vec M}\Longleftrightarrow\vec y\in\xdown f{\vec M}.
  \end{equation}
\end{lemma}

\begin{proof}
  Assume first that $\vec y\in \Pdown{P}{\vec M}$. Then, there exists
  $\vec{x}\in \vec M$ such that $\vec{y}\leq_P \vec{x}$. We prove that
  $\vec{x}\in \filter{\vec M}f$ by observing that if $i$ is an index
  such that $\vec{f}(i)<\omega$, then $i\in P$ and
  $\vec{f}(i)=\vec{y}(i)<\omega$. From $i\in P$ we get
  $\vec{x}(i)=\vec{y}(i)$. Hence $\vec{x}(i)=\vec{f}(i)$ and we have
  proved that $\vec{x}\in \filter{\vec M}f$. Since $\vec{y}\leq \vec{x}$,
  we get $\vec{y}\in \xdown{f}{\vec M}$.

  \smallskip Conversely, assume that $\vec y\in\xdown f{\vec M}$: there
  exists $\vec{x}\in \filter{\vec M}f$ such that $\vec{y}\leq \vec{x}$. Let
  $i\in P$. If $\vec{y}(i)=\omega$ then from $\vec{y}(i)\leq \vec{x}(i)$ we get
  $\vec{y}(i)=\vec{x}(i)$. If $\vec{y}(i)<\omega$ then
  $\vec{f}(i)=\vec{y}(i)$ and form $\vec{x}\in \filter{\vec M}f$ we get
  $\vec{x}(i)=\vec{f}(i)$. Hence in both cases, we have
  $\vec{x}(i)=\vec{y}(i)$. We have proved that
  $\vec{y}\leq_P\vec{x}$. Therefore $\vec{y}\in \Pdown{P}{\vec M}$.
\end{proof}

Let us now introduce another set, again for a set $\vec M\subseteq\setN^d_\omega$:
\begin{align*}
  \Pyset{M}&=\bigl\{(P,\vec{y})\in \partie{\{1,\ldots,d\}}\times\setN_\omega^d \mid \vec{y}\in\Pdown{P}{\vec M}\bigr\}
\end{align*}

\begin{corollary}[of Lemma~\ref{lem:fp}]
  \label{cor:f-vs-P}
  The membership problems in $\Pyset{M}$ and in $\fyset{M}$ are
  inter-reducible.  Both reductions are effective: from an algorithm
  solving the first problem, we construct an algorithm solving the
  second one.
\end{corollary}

\begin{proof}
  From $P\subseteq\{1,\ldots,d\}$ and $\vec y\in\setN^d_\omega$, define $\vec{f}\in\setN_\omega^d$
  by~\eqref{eq:11}. From \eqref{eq:a:1}, we deduce that $(P,\vec y)\in
  \Pyset{M}$ if and only if $(\vec f,\vec y)\in \fyset{M}$.

  \smallskip\noindent Conversely, let $\vec{f}\in\setN_\omega^d$ and $\vec
  y\in\setN^d_\omega$. Observe that if $\vec{y}\not\leq\vec{f}$ then $\vec{y}\not\in
  \xdown f{\vec M}$. So we can assume that $\vec{y}\leq \vec{f}$. We
  introduce the set $P=\bigl\{i\in\{1,\ldots,d\}\mid \vec{f}(i)<\omega\bigr\}$ and the
  vector $\vec{z}\in\setN_\omega^d$ defined by $\vec{z}(i)=\vec{f}(i)$ if $i\in P$
  and $\vec{z}(i)=\vec{y}(i)$ otherwise. We have $\vec{y}\in \xdown
  f{\vec M}$ if and only if $\vec{z}\in \xdown f{\vec M}$. Moreover,
  from Lemma~\ref{lem:fp} we deduce that $\vec z\in\xdown f{\vec M}$ if
  and only if $\vec z\in \Pdown{P}{\vec M}$. In summary, $(\vec y,\vec
  f)\in\fyset M$ if and only if $\vec y\leq\vec f$ and $(\vec z,P)\in\Pyset
  M$.
\end{proof}

To establish Theorem~\ref{thm:RE}, it remains, in view of
Lemma~\ref{lem:vj-down} and Corollary~\ref{cor:f-vs-P}, to find an
algorithm solving membership to $\Pyset{{}\Lim\poststar(\VV)}$.  This
is obtained by first proving that, for a VAS $\VV_P$ suitably
constructed from $\VV$ and $P$, we have
 \begin{equation}
   \label{eq:12}
   \poststar(\VV_P)=\xcover{\leq_P}(\VV)
\end{equation}
which implies $\Lim\poststar(\VV_P)=
\Lim\xcover{\leq_P}(\VV)=\Lim\Pdown{P}\poststar(\VV)=\Pdown{P}\Lim\poststar(\VV)$. Then,
Theorem~\ref{thm:reclim} applied to $\VV_P$ will give an algorithm to
decide membership in this set. Since there is a finite number of
subsets $P$ of $\{1,\ldots,d\}$, this yields an algorithm to decide
membership in $\Pyset{{}\Lim\poststar(\VV)}$.

\medskip So let $\VV = \tup{A,\delta,\ini}$ be a VAS and $P\subseteq\{1,\ldots,d\}$, and
let us define a VAS $\VV_P$ satisfying~\eqref{eq:12}. We consider $d$
distinct additional elements $b_1,\ldots,b_d\not\in A$. Let
$B=\{b_1,\ldots,b_d\}$. We consider the VAS $\VV_{P}=\tup{A\uplus B,\delta_P,\ini}$,
where $\delta_P$ extends $\delta$ by:
\begin{equation*}
  \delta_P(b_i)=\begin{cases}
    \vec0 & \text{ if $i\in P$}\\
    -\vec{e_i} & \text{ if $i\notin P$.}
  \end{cases}
\end{equation*}

\begin{lemma}\label{lem:redP}
  Let $\VV_{P}$ constructed from $\VV$ and $P$ as above. Then 
  $\xcover{\leq_P}(\VV) = \poststar(\VV_P) $.
\end{lemma}

\begin{proof}
  Let $\vec{x}\in \xcover{\leq_P}(\VV)$. By definition, there exists
  $\vec{y}\in \poststar(\VV)$ such that $\vec{x}\leq_P\vec{y}$. Note that
  $\vec y\in\poststar(\VV_P)$, and that $\vec{y}\xrightarrow{u}\vec{x}$
  in $\VV_P$ with $u=\prod_{i=1}^db_i^{\vec{y}(i)-\vec{x}(i)}$, so
  $\vec{x}\in \poststar(\VV_P)$. Conversely let $\vec{x}\in
  \poststar(\VV_P)$, and $u\in(A\cup B)^*$ such that
  $\ini\xrightarrow{u}_{\VV_P}\vec{x}$. Let~$v$ be obtained from $u$
  by erasing all letters of $B$. Since $\delta_P(b)\leq\vec0$ for $b\in B$, the
  word $v$ is fireable from $\ini$. Thus
  $\vec{y}=\ini+\delta(v)\in\poststar(\VV)$.  By definition of $\VV_P$ we
  have $\vec{x}\leq_P\vec{y}$, so~$\vec{x}\in \xcover{\leq_P}(\VV)$.
\end{proof}

As explained above, Theorem~\ref{thm:RE} is now established, by
combining Lemmas~\ref{lem:vj-down} and Corollary~\ref{cor:f-vs-P}
applied to $\vec M=\Lim\poststar(\VV)$, as well as
Lemma~\ref{lem:redP}.

\section{Computing the cover of a VAS with one zero-test}
\label{sec:alg}

This section describes an algorithm computing a basis of the cover of a \vasz given as input.

\medskip It will be convenient to consider \vas or \vasz whose initial
state belongs to $\setN^d_\omega$. The semantics given by~\eqref{eq:vasz} is
generalized by extending addition to $\setN_\omega$, letting $\omega+n=n+\omega=\omega$ for
all $n\in\setZ$. Notice that all results obtained so far for a \vas, and in
particular Theorem~\ref{thm:RE}, extend to \vas with such generalized
initial states. Indeed, an $\omega$ value in some component of \ini remains
frozen to $\omega$, whatever action is executed, and can therefore be
safely ignored.

We introduce a notation to change the initial state of a
\vas/\vasz~\VV. For   $\vec x \in \setN^d_\omega$, 
 we let $\VV(\vec
x)$ be the \vas/\vasz obtained from $\VV$ by replacing the initial state $\ini$ by~$\vec x$. 

\medskip In this section, we fix a \vasz
$\VV_z=\tup{A,\az,\delta,\ini}$. To simplify the presentation, we assume
without loss of generality that $\ini\in\{0\}\times\setN^{d-1}$, and that $\delta(\az) \in
\{0\} \times \setZ^{d-1}$.  In the sequel, we denote by $\VV=\tup{A,\delta,\ini}$ the
VAS obtained from $\VV_z$ by removing the zero test.  We shall work
with a single filter throughout the section: we introduce
$\vec{f}=(0,\omega,\ldots,\omega)$.

\medskip
\paragraph{\bfseries Input/output of the algorithm}
Our algorithm is inspired by Karp and Miller's one for a
VAS~\cite{Karp&Miller:69}.  Given as input a \vasz $\VV_z$, it builds
a finite tree with nodes labeled by vectors in $\{0\} \times \setN_\omega^{d-1}$, such
that when the algorithm terminates:
\begin{equation}
  \label{eq:13}
  \text{The set $\vec R$ of node labels is a basis of~$\xdown{f}{\poststar(\VV_z)}$.}\tag{$\ast$}
\end{equation}
Observe that, at the end of the algorithm, $\vec R$ is not a basis of
the whole cover of $\VV_z$, but only a basis of an $\vec f$-filtered
cover of $\VV_z$.

Let us first explain how to compute from $\vec R$ a basis of
$\cover{}(\VV_z)$. If $\vec x\in\cover{}(\VV_z)$, then there exist $u\in
A^*$ and $\vec y\in\setN^d$ such that $\ini\xrightarrow{u}\vec y\geq\vec x$.
Let us factorize $u=u_1u_2$, where $u_1$ ends with the last zero test
\az, or is empty if there is no zero-test. Then, we have
$\ini\xrightarrow{u_1}\vec r\xrightarrow{u_2}\vec y\geq\vec x$, with
$\vec r\in\{0\}\times\setN^{d-1}$ (if $u_1$ is empty, we use the assumption
$\ini\in\{0\}\times\setN^{d-1}$). In particular, $\vec
r\in\xdown{f}{\poststar(\VV_z)}=\dc \vec R\cap\setN^d$ by~\eqref{eq:13}. Since
no zero-test occurs in $u_2$, the state $\vec y$ reached after firing
$u$ belongs to $\poststar({\VV}(\vec r))$, and therefore, $\vec
x\in\dc\poststar(\VV(\vec r))$. This simple remark yields the following
result:
\begin{lemma}
  \label{lem:cover-from-filter} If $\vec R$ is a basis of $\xdown
  f\poststar(\VV_z)$, then $\cover{}(\VV_z)=\bigcup_{\vec r\in\vec
    R}\dc\poststar(\VV(\vec r))$.
\end{lemma}
In words, we obtain a basis of $\cover{}(\VV_z)$ as the union of all
bases output by the usual Karp-Miller algorithm run on inputs
$\VV(\vec r)$, for $\vec r\in\vec R$. Let us now explain how to
compute~$\vec R$.

\paragraph{\bfseries Outline of the algorithm.}
To build a tree whose set of labels is $\vec R\subseteq\{0\}\times\setN_\omega^d$, the
algorithm works top-down from the root labeled by the initial state
$\ini\in\{0\}\times\setN^{d-1}$. Its main loop is similar to that of the
Karp-Miller algorithm: for each leaf of the tree,

\begin{enumerate}[leftmargin=*]
\item if the label of the leaf already occurs above it along the path
  to the root, then the leaf is not expanded, and will remain a leaf
  during the execution of the algorithm.
\item\label{item:6} Otherwise, we try to expand the tree from the
  leaf. As in the Karp-Miller algorithm:
  \begin{enumerate}[label=$\alph*.$,ref=$\alph*$,leftmargin=*]
  \item\label{item:7} we perform some standard acceleration, which is
    explained below,

  \item\label{item:8} we then expand the leaf, adding new children to
    it. However, unlike the Karp-Miller algorithm, which fires all
    original transitions of the VAS from the label of the leaf, we add
    two kinds of children to the current leaf labeled $\vec
    x\in\{0\}\times\setN^{d-1}_\omega$:
    \begin{enumerate}[label=$(\roman*)$,leftmargin=*]
    \item\label{item:9} one child corresponding to firing the
      zero-test from the leaf label, if possible,

    \item\label{item:10} several children representing a basis of
      $\xdown f{\poststar({\VV(\vec x)})}$.
    \end{enumerate}
  \end{enumerate}
\end{enumerate}
Note that Step~\ref{item:10} involves \VV and not $\VV_z$, \ie, the
zero-test is not considered during this step. It is a macro-step
computing itself a basis of a cover, to be used in the whole
computation.  In the particular case where the \vasz is obtained by
just adding to states of a VAS an extra first component, left
untouched (therefore remaining 0 forever) and where the zero-test is
never fired, step~\ref{item:10} actually computes in one shot the
cover of the original VAS (completed with the first component, left
to~0).  Theorem~\ref{thm:RE} shows that Step~\ref{item:10} is
effective.


\medskip We now enter the details of the algorithm.  At any step of
the execution, in the tree built by the algorithm, every ancestor node
$n_{\vec x}$ of a node $n_{\vec y}$ satisfies the invariant $\invar
{\vec{x}} {\vec{y}}$ where $\vec{x},\vec{y}$ are the labels of
$n_{\vec x},n_{\vec y}$ and where $\trans$ is the binary relation
defined over $\{0\}\times \setN_\omega^{d-1}$ by:
\begin{equation*}
  \invar {\vec{x}} {\vec{y}}~\text{ if }~\vec{y} \in\xdown f\Lim\poststar(\VV_z(\vec{x})).
\end{equation*}
By the next lemma, it is sufficient to maintain this invariant along
each parent-child edge.

\begin{lemma}
  \label{lem:transitive}
  The binary relation $\trans$ over $\{0\}\times\setN^{d-1}_\omega$ is reflexive and transitive.
\end{lemma}

The proof of Lemma~\ref{lem:transitive} is itself based on the
following intermediate statement. To shorten notation, for a set $\vec
M\subseteq\setN^d_\omega$, we let $\postset\vec M=\bigcup_{\vec x\in\vec M}\poststar(\VV_z(\vec
x))$ denote the set of states that can be reached in $\VV_z$ from any
initial vector chosen in $\vec M$ (in this notation used only in
Lemmas~\ref{lem:transitive} and \ref{lem:limreachlim}, the \vasz will
always be $\VV_z$, and is therefore omitted).
\begin{lemma}
  \label{lem:limreachlim}
  Let $\vec M\subseteq\setN_\omega^d$. Then, we have
  \begin{math}
    \Lim\postset\Lim\vec M=\Lim\postset\vec M.
  \end{math}

\end{lemma}

\begin{proof}
  Since $\vec M\subseteq\Lim \vec M$, we have $\Lim\postset\vec
  M\subseteq\Lim\postset\Lim\vec M$. For the other inclusion, pick $\vec
  x\in\Lim\postset\Lim\vec M$. This means that we have the following
  situation
  \begin{equation*}
    \vec y_n \limra \vec y \xrightarrow{u_n} \vec x_n \limra \vec x,
  \end{equation*}
  with $\vec y_n\in\vec M$, $\vec y,\vec x_n\in\setN_\omega^d$ and $u_n\in A^*$ for
  all $n$.

  \smallskip Since $\lim_n\vec y_n=\vec y$, we may assume that $\vec
  y_n(i)=\vec y(i)$ for all $n$ if $\vec y(i)<\omega$, and that $(\vec
  y_n(i))_n$ is strictly increasing if $\vec y(i)=\omega$.  Let $k_n$ be a
  strictly increasing sequence such that $k_n\geq n+\max_{1\leq i\leq
    d}|\delta(u_n)(i)|$, and let $\vec y'_n=\vec y_{k_n}$. Clearly,
  $\lim_n\vec y'_n=\vec y$. By construction, $u_n$ is fireable from
  $\vec y'_n$: let $\vec y'_n\xrightarrow{u_n}\vec x'_n$. We then have
  $\vec x'_n(i)=\vec x_n(i)$ if $\vec y(i)<\omega$, and $\vec x'_n(i)\geq n$
  if $\vec y(i)=\omega$. So, $\vec x=\lim_n\vec x'_n\in\Lim\postset\vec M$.
\end{proof}

\begin{proof}[Proof of Lemma~\ref{lem:transitive}]
  Reflexivity is obvious. For transitivity, assume that $\invar {\vec
    x}{\invar{\vec y}{\vec z}}$. Then by definition of $\trans$, we
  have $\vec{z}\in\xdown{f}\Lim\poststar(\VV_z(\vec y))$ and $\vec{y} \in
  \xdown{f}\Lim\poststar(\VV_z(\vec x))$.  Since $\vec f=(0,\omega,\ldots,\omega)$,
  we can use monotony to obtain $\xdown{f}\poststar(\VV_z(\vec
  x))=\poststar\xdown{f}\poststar(\VV_z(\vec x))$. We deduce from this
  equality that
  \begin{alignat*}{2}
    \xdown{f}\Lim\xdown{f}\poststar(\VV_z(\vec x))&=\xdown{f}\Lim\poststar\xdown{f}\poststar(\VV_z(\vec x))&&\text{\quad by applying the monotonous}\\[-.5ex]
    &&&\text{\quad operator $\xdown{f}\Lim$,}\\
    &=\xdown{f}\Lim\poststar\Lim\xdown{f}\poststar(\VV_z(\vec
    x))&&\text{\quad by Lemma~\ref{lem:limreachlim}}.
  \end{alignat*}
  Since $\Lim$ and $\xdown{f}$ commute (see \eqref{eq:10}), and since
  the operator $\xdown{f}$ is obviously idempotent, we finally get
  $\xdown{f}\Lim\poststar(\VV_z(\vec
  x))=\xdown{f}\Lim\poststar\xdown{f}\Lim\poststar(\VV_z(\vec
  x))$. Now, the hypotheses imply that $\vec
  z\in\xdown{f}\Lim\poststar\xdown{f}\Lim\poststar(\VV_z(\vec x))$. We
  deduce that $\vec{z}\in\xdown{f}\Lim\poststar(\VV_z(\vec x))$, that
  is,~$\invar{\vec x}{\vec z}$.
\end{proof}

\medskip

Assume now that $\vec{x}\in \{0\}\times\setN_\omega^{d-1}$ labels a leaf. We create a
child of this leaf if the vector $\vec{y}=\vec{x}+\delta(\az)$ is
nonnegative. Note that in this case $\vec{y}\in\{0\}\times\setN_\omega^{d-1}$, since
$\delta(\az)(1)=0$. We do not violate the invariant when creating the child
labeled~$\vec{y}$ since $\invar{\vec{x}}{\vec{y}}$. We also add new
children labeled by elements of the minimal basis $\vec B(\vec{x})$ of
$\xdown f\Lim\poststar(\VV(\vec{x}))$. Since $\setN^d\cap \xdown
f\Lim\poststar(\VV(\vec{x}))$ is equal to $\setN^d\cap \xdown
f\poststar(\VV(\vec{x}))$, by Theorem~\ref{thm:RE}, one can compute
$\vec B(\vec x)$. Observe that $\invar{\vec{x}}{\vec{b}}$ for every
$\vec{b}\in \vec B(\vec{x})$, so that the invariant is still fulfilled
after adding elements of $\vec B(\vec x)$.

\bigskip
The termination of the algorithm is obtained by introducing an
\emph{acceleration operator}~$\nabla$.  For
$\vec{x},\vec{y}\in \{0\}\times\setN_\omega^{d-1}$ such that $\vec x\leq\vec y$, we define the
vector $\vec{x}\nabla\vec{y}\in\{0\}\times\setN^{d-1}_\omega$~by:
$$
(\vec{x}\nabla\vec{y})(i)=
\begin{cases}
  \omega & \text{if\/ $\vec{ x}(i) < \vec {y}(i)$} \\
  \vec x(i) & \text{if\/ $\vec{ x}(i) = \vec {y}(i)$}.
\end{cases}
$$
Let us first verify that performing acceleration cannot violate the invariant.
\begin{lemma} \label{lemm:accel-sound} If $\invar{\vec{x}}{\vec{y}}$
  with $\vec{x}\leq \vec{y}$ then $\invar{\vec{x}}{(\vec x\nabla\vec y)}$.
\end{lemma}

\begin{proof}
  If  $\invar{\vec{x}}{\vec{y}}$, then $\vec y\in\xdown
  f\Lim\poststar(\VV_z(\vec x))$, and  we obtain the following situation
  \begin{equation*}
    \vec x\xrightarrow{~u_n~}\vec z_n \limra{\vec z}\geq\vec y,
  \end{equation*}
  with $u_n\in (A \cup \{ \az \})^*$ and $\vec z,\vec
  z_n\in\{0\}\times\setN_\omega^{d-1}$. Since $\vec z\geq\vec y\geq\vec x$, there exists $\ell$
  such that $\vec z_\ell(i)\geq\vec x(i)$ for all indices $i$ satisfying
  $\vec x(i)<\omega$, and further $\vec z_\ell(i)>\vec x(i)$ if $\vec
  x(i)<\vec y(i)$.
  Therefore, $\vec z_\ell\geq\vec x$, and as we have $\vec z_\ell(1) = \vec
  x(1) = 0$, we deduce that $u_\ell^k$ is fireable from $\vec x$ for
  all~$k$. Call $\vec t_k$ the state reached from $\vec x$ after
  firing $u_\ell^k$. Then we have $\vec t_k\in\{0\}\times\setN^{d-1}_\omega$ and
  $\lim_{k\to\infty}\vec t_k\geq\vec x\nabla\vec y$, which proves $\vec x\nabla\vec
  y\in\xdown f\Lim\poststar(\VV_z(\vec x))$.
\end{proof}

\begin{algorithm}[htpb]
  \caption{An algorithm to compute a basis of $\xdown{f} \poststar(\VV_z)$}
  \label{alg:main}
  \begin{itemize}
  \item \underbar{\sf Inputs}: A \vasz $\VV_z$ such that
    $\ini\in\{0\}\times\setN^{d-1}$ and $\delta(\az)\in\{0\}\times\setZ^{d-1}$.
  \item \underbar{\sf Outputs}: $\vec R$, a finite subset of $\{0\} \times 
    \setN^{d-1}_\omega$.
  \item \underbar{\sf Internal Variables}: 
    \begin{itemize}
    \item $\TT$, a tree labeled by elements of $\setN^d_\omega$.
    \item $\NN$, a set of nodes.
    \end{itemize}
  \item \underbar{\sf Algorithm}:
    \begin{algorithmic}[1]
      \State {Initialize $\TT$ as a single root $n_{\textit{in}}$, labeled by $\ini$}
      \State {$\NN \gets \{ n_{\textit{in}} \}$}
      \While {$\NN \neq \emptyset$}
      \State {Choose a node $n$ from $\NN$}
      \State {$\NN\gets \NN\setminus\{n\}$}
      \State {$\vec x \gets label(n)$}
      \If{no strict ancestor of $n$ has label $\vec{x}$}     
      \ForAll {strict ancestor $n_0$ of $n$} \label{line:accel-begin}
      \Comment{\small\smaller Acceleration, step~\ref{item:6}.\ref{item:7}}
      \State {$\vec{x}_0 \gets label(n_0)$}
      \If {$\vec{x}_0\leq \vec{x}$}
      \State{$\vec{x}\gets\vec{x}_0\nabla\vec{x}$}
      \EndIf \label{line:accel-end}
      \EndFor
      \State{Replace the label of $n$ by $\vec{x}$}
      \If{$\vec{x}+\delta(\az)\geq \vec0$} \label{line:az}
      \Comment{\small\smaller Expand by zero-test, step~\ref{item:6}.\ref{item:8}~\ref{item:9}}
      \State {Create a new node in $\TT$ labeled by $\vec{x}+\delta(\az)$, as 
        a child of $n$} \label{line:create-az}
      \State {Add this node to $\NN$}
      \EndIf
      \ForAll{$\vec{b}\in \vec B(\vec{x})$} \label{line:vas}
      \Comment{\small\smaller Expand by $\vec B(\vec x)$, step~\ref{item:6}.\ref{item:8}~\ref{item:10}}
      \State {Create a new node in $\TT$ labeled by $\vec{b}$, 
        as a child of $n$} \label{line:create-vas}
      \State {Add this node to $\NN$}
      \EndFor
      \EndIf
      \EndWhile
      \State {$\vec R \gets \bigl\{ label(n) ~|~ n \in nodes(\TT) \bigr\}$}
      \State{\Return {$\vec R$}}
    \end{algorithmic}
  \end{itemize}
\end{algorithm}

Algorithm \ref{alg:main} computes $\vec R$.  If every leaf has a
(strict) ancestor with the same label, then it terminates and returns
the current set of node labels. If it finds some leaf $n$ whose
ancestors carry different labels than that of~$n$, it performs
acceleration at $n$ (step \ref{item:6}.\ref{item:7} of the outline):
while $n$ has an ancestor $n_0$ labeled by a vector $\vec{x_0}$ such
that $\vec{x_0}\leq\vec{x}<\vec{x_0}\nabla\vec{x}$, it replaces the label
$\vec{x}$ of the leaf $n$ with $\vec{x_0}\nabla\vec{x}$. 

From Lemma~\ref{lemm:accel-sound}, we deduce that the invariant still
holds. Since this loop just replaces some components by $\omega$, it
terminates.  Finally, once the label $\vec x$ of $n$ has been updated,
the algorithm creates a new child labeled by $\vec{x}+\delta(\az)$ if this
vector is nonnegative (step~\ref{item:6}.\ref{item:8}\ref{item:9}),
and it creates a new child of $n$ labeled by $\vec{b}$ for each
$\vec{b}\in \vec B(\vec{x})$ (step
\ref{item:6}.\ref{item:8}\ref{item:10}).  Note that all labels belong
to $\{0\}\times\setN^d_\omega$, since $\{\xini,\delta(\az)\}\cup \vec B(\vec x)\subseteq\{0\}\times\setN^d_\omega$.
\begin{proposition}
  \label{prop:algo-correctness}
  Algorithm \ref{alg:main} terminates, and it returns a finite set
  $\vec R$ such that
  \begin{equation}
    \label{eq:14}
    \dc \vec R =\xdown{f}{\Lim\poststar(\VV_z)}.
  \end{equation}
\end{proposition}

\begin{proof}
  The termination of the algorithm follows from K\"onig's lemma. If
  the algorithm does not terminate, then it would generate an infinite
  tree. Because this tree has a finite branching degree, by K\"onig's
  lemma, there is an infinite branch. Since $\leq$ is a well-ordering over
  $\{0\}\times\setN_\omega^{d-1}$, this implies that we can extract from this infinite
  branch an infinite increasing subsequence. However, since we add
  children to a leaf only if there does not exist a strict ancestor
  labeled by the same vector, this sequence cannot contain the same
  vector twice, and must therefore be \emph{strictly} increasing. But,
  due to the use of the operator $\nabla$, a component with an integer is
  replaced by $\omega$ at every acceleration step. Because the number of
  $\omega$'s in the vectors labeling a branch cannot decrease, we obtain a
  contradiction. Let us now prove~\eqref{eq:14}.

  \begin{itemize}[leftmargin=1.7em]
  \item [$\subseteq$] Let $n$ be a node of $\TT$, whose label is $\vec x$. By
    Lemmas~\ref{lem:transitive} and \ref{lemm:accel-sound}, we have
    $\invar{\vec{x_{\textit{in}}}}{\vec x}$. By definition of
    $\trans$, we conclude that $\vec x\in \xdown{f}
    \Lim\poststar(\VV_z)$.

    \smallskip
  \item [$\supseteq$] We shall show $\xdown{f}{\poststar(\VV_z)}\subseteq\dc\vec
    R$. The desired inclusion follows by taking limits of both sides,
    since
    $\Lim\xdown{f}{\poststar(\VV_z)}=\xdown{f}\Lim{\poststar(\VV_z)}$
    and $\Lim\dc\vec R=\dc\vec R$ (since $\vec R$ is finite).  So let
    $(0, \vec \alpha) \in \xdown{f} \poststar(\VV_z)$: there exist $\vec {\alpha}'
    \in \setN^{d-1}$ with $\vec\alpha\leq\vec{\alpha}'$ and $u \in (A \cup \{ \az \})^*$ such
    that $\ini \ru{u} (0, \vec {\alpha}')$. We will show by induction on
    the length of $u$ that $(0, \vec {\alpha}') \in \dc \vec R$. If $u$ is
    empty, just observe that $\ini$ labels the root, hence $\ini\in \vec
    R$. Otherwise, $u=va$ and we have:
    $$
    \ini \ru{v} (0, \vec \beta) \ru{a} (0, \vec {\alpha}')
    $$

    The induction hypothesis yields $(0, \vec \beta) \in \dc \vec R$. Hence,
    there is in the tree a node labeled $\vec \gamma\geq\vec \beta$. Since a node
    label cannot be modified after acceleration (lines
    \ref{line:accel-begin} to \ref{line:accel-end}), this means that
    instructions at lines \ref{line:az} and \ref{line:vas} have been
    executed when the variable $\vec x$ was set to $\vec \gamma$, and this
    ensures that $\vec {\alpha}' \in \dc \vec R$.

  \end{itemize}    

  \noindent
  We have proved that Algorithm~\ref{alg:main} computes a basis $\vec R$ of
  $\xdown{f}{\poststar(\VV_z)}$. 
\end{proof}

Proposition~\ref{prop:algo-correctness} and
Lemma~\ref{lem:cover-from-filter} finally imply the central theorem of
this paper:

\begin{theorem}
  Given a \vasz $\VV_z$, one can effectively compute the minimal basis of
  $\cover{}(\VV_z)$.\label{thm:cover-vas0}
\end{theorem}

This theorem solves the place-boundedness problem for \vasz. For
vector addition systems, it can be transferred to obtain
model-checking algorithms. We investigate model-checking problems in
the presence of one zero-test in the next section. However, we shall
use the decidability of the reachability problem instead of
Theorem~\ref{thm:cover-vas0}.

\section{Repeated Control State Reachability is decidable for \texorpdfstring{\vassz}{VASS0}}
\label{sec:rcsrp}

Vector addition systems can be extended with control flow graphs. Such
a control flow graph is given by a finite set of \emph{control states}
and a finite set of \emph{transitions} labeled by \emph{actions}. This
model is called \emph{Vector Addition Systems with States} (\vass for
short). If instead of a VAS, we enrich a \vasz with a control flow
graph, we obtain a \emph{Vector Addition System with States and one
  zero-test} (\vassz for short). These models are formally defined in
the sequel.

For these systems, the repeated control state reachability consists in
deciding whether a given control state can be visited infinitely often
along some run. This problem is interesting since a number of
model-checking problems, such as LTL model-checking, are reducible to
it. For the class of \vass, the repeated control state reachability
problem is known to be decidable thanks to a reduction to the
\emph{computation of the cover set}. In this section, we extend this
decidability result for the class of \vassz. However, our proof relies
on a reduction to the \emph{reachability problem} for
\vassz~\cite{Reinhardt:08,MFCS:11}. We leave as an open question
whether the repeated control state reachability for \vassz can be
reduced to the computation of the cover.

\medskip

Let us first recall the classical extensions of \vas and \vasz with
States, respectively written \vass and \vassz. States can be seen as
mutually-exclusive, 1-bounded counters, and hence are only used as a
syntactic \hbox{convenience}.
\begin{definition}\textbf{(\vassz)}
  A \emph{Vector Addition System with States and one zero-test
    (\vassz) of dimension $d$} is a tuple $\VV = \tup{A, \az,
    \delta,\ini,Q,T, \qini}$, where $\tup{A, \az, \delta,\ini}$ is a \vasz of
  dimension $d$, $Q$ is a non-empty finite set of \emph{control
    states}, $T\subseteq Q\times (A \cup \{\az\})\times Q$ is a finite set of
  \emph{transitions}, and $\qini \in Q $ is the \emph{initial control
    state}.
\end{definition}
A \emph{Vector Addition System with States (\vass)} is defined
similarly from a \vas $\tup{A, \delta,\ini}$, with $T\subseteq Q\times A\times Q$, and can be
thought of as a \vassz where the action \az is not used. The \vassz
semantics is defined as follows. Let us call \emph{state} any pair
$(q,\vec x)\in Q\times\setN^d$. A \vassz of dimension $d$ induces a transition
system over the set of states, given for every $a\in A\cup\{\az\}$ by:
\begin{equation*}
  (p,\vec{x}) \ru{~a~} (q, \vec{y}) \text{ if }  (p,a,q)\in T\text{ and }\vec{x}\ru{~a~} \vec{y}
\end{equation*}
These relations extend uniquely into relations $\ru{~w~}$ over the set
of states, for $w\in (A\cup\{\az\})^*$, by requiring that $\ru{~\varepsilon~}$ is the
identity relation and $\ru{~w_1w_2}$ is the composition
$\ru{~w_1~}\circ\ru{~w_2~}$, for $w_1,w_2\in (A\cup\{\az\})^*$. The
\emph{reachability relation}, denoted by $\ru{~*~}$ is defined as the
union of all relations $\ru{~w~}$, when $w$ ranges over
$(A\cup\{\az\})^*$. We also introduce the relation $\ru{~+~}$ defined as
the union of all relations $\ru{~w~}$ when $w$ ranges over
$(A\cup\{\az\})^+$.

\medskip

A control state $\qf\in Q$ is said to be \emph{visited infinitely often}
if there exists an infinite sequence $(\vec{x}_j)_{j>0}$ of vectors
$\vec{x}_j\in \setN^d$ such that $(\qini,\ini)\ru{~*~}(\qf,\vec{x}_1)$ and
such that $(\qf,\vec{x}_{j})\ru{~+~}(\qf,\vec{x}_{j+1})$ for all
$j>0$. The \emph{repeated control state reachability} consists in
deciding whether a given control state $\qf$ is visited infinitely
often.

\medskip

We first reduce the repeated control state reachability to a simpler
property.

\begin{lemma}\label{lem:repeat}
  Let $\VV = \tup{A, \az, \delta,\ini,Q,T, \qini}$ be a \vassz of
  dimension~$d$.  A control state $\qf$ is visited infinitely often if
  and only if there exist $\vec x,\vec y\in\setN^d$ such that
  $(\qini,\ini)\ru{~*~}(\qf,\vec{x})\ru{~w~}(\qf,\vec{y})$, and one of
  the following conditions is satisfied:
  \begin{enumerate}[label=$(\roman*)$]
  \item\label{item:11} we have $\vec{x}\leq \vec{y}$ and $w\in A^+$, or
  \item\label{item:12} we have $\vec{x}\leq_1\vec{y}$ and $w\in (A\cup\{\az\})^+$.
  \end{enumerate}
\end{lemma}

\begin{proof}
  Naturally, if \ref{item:11} or \ref{item:12} holds, then $\qf$ is
  visited infinitely often by monotony of $\ru{~w~}$.  Conversely,
  assume that $\qf$ is visited infinitely often. There exists an
  infinite sequence $(\vec{x}_j)_{j>0}$ of vectors $\vec{x}_j\in \setN^d$, a
  word $w_0\in (A\cup\{\az\})^*$ such that
  $(\qini,\ini)\ru{~w_0~}(\qf,\vec{x}_1)$, and an infinite sequence
  $(w_j)_{j>0}$ of words $w_j\in (A\cup\{\az\})^+$ such that
  $(\qf,\vec{x}_{j})\ru{~ w_j ~}(\qf,\vec{x}_{j+1})$ for every
  $j>0$. We introduce the set $J$ of indexes $j>0$ such that $\az$
  occurs in $w_j$. We distinguish two cases according to whether $J$
  is finite or infinite.

  Assume first that $J$ is finite. By replacing $w_0$ with $w_0\cdots w_m$,
  where $m=\max J$, and $w_\ell$ with $w_{m+\ell}$ for $\ell>0$, we may assume
  without loss of generality that $J=\emptyset$, \ie, that $w_j\in A^+$ for all
  $j>0$. By Dickson's lemma, there exist positive integers $j<k$ such
  that $\vec{x}_j\leq \vec{x}_k$. We deduce that~\ref{item:11} holds, by
  observing that
  $(\qini,\ini)\ru{~v~}(\qf,\vec{x})\ru{~w~}(\qf,\vec{y})$ with
  $v=w_0\ldots w_{j-1}$, $w=w_{j}\ldots w_{k-1}$, $\vec{x}=\vec{x}_j$ and
  $\vec{y}=\vec{x}_k$.

  Assume now that $J$ is infinite. By suitably concatenating some
  words $w_j$, we can assume without loss of generality that $\az$
  occurs in $w_j$ for every $j>0$. This means that $w_j$ can be
  decomposed into $w_j=u_j \az v_j$ for some words $u_j,v_j\in
  (A\cup\{\az\})^*$. Hence there exists a state $(q_j,\vec{y}_j)$ such that
  $(\qf,\vec{x}_{j})\ru{~u_j \az~}(q_j,\vec{y}_j)\ru{~
    v_j~}(\qf,\vec{x}_{j+1})$. Dickson's lemma shows that there exist
  $j<k$ such that $\vec{y}_j\leq \vec{y}_k$ and $q_j=q_k$. Since the
  vectors $\vec{y}_j$ and $\vec{y}_k$ appear just after the zero test
  $\az$, we deduce that $\vec{y}_j(0)=\vec{y}_k(0)$, so
  $\vec{y}_j\leq_1\vec{y}_k$. Let $\vec{z}=\vec{y}_k-\vec{y}_j$. Note
  that we have:
  $$(\qini,\ini)\ru{~w_0\ldots w_{j-1}u_j\az}(q_j,\vec{y}_j)
  \xrightarrow{~v_j~} (\qf,\vec{x}_{j+1})\xrightarrow{~w_{j+1}\ldots
    w_{k-1}u_k \az~}(q_k,\vec{y}_k)$$ Now we use monotony: since
  $(q_j,\vec{y}_j)\xrightarrow{~v_j~}(\qf,\vec{x}_{j+1})$,
  $\vec{y}_j\leq_1\vec{y}_k$, and $q_k=q_j$, we get
  $(q_k,\vec{y}_k)\xrightarrow{~v_j~}(\qf,\vec{x}_{j+1}+\vec{z})$. Therefore
  $(\qini,\ini)\ru{~v~}(\qf,\vec{x})\ru{~w~}(\qf,\vec{y})$ with
  $v=w_0\ldots w_j$, $w=w_{j+1}\ldots w_{k-1}u_k \az v_j$,
  $\vec{x}=\vec{x}_{j+1}$, and $\vec{y}=\vec{x}_{j+1}+\vec{z}$.
\end{proof}

\begin{theorem}
  \label{theo:rcsrp}
  The repeated control state reachability problem is decidable for \vassz.
\end{theorem}

\begin{proof}
  Consider a \vassz $\VV = \tup{A, \az, \delta,\ini,Q,T,\qini}$ of
  dimension $d$ and a control state $\qf\in Q$. Without loss of
  generality, by introducing some extra control states and actions, we
  can assume that $\delta(\az)$ is the zero vector.

  We construct from $\VV$ a \vassz $\VV'=\tup{A',\az,\delta',Q',T',\qini}$
  of dimension $2d$ as follows. We duplicate the set of control states
  $Q$ into two additional copies for simulating
  conditions~\ref{item:11} and \ref{item:12} of
  Lemma~\ref{lem:repeat}. These copies are denoted by $Q_{(i)}$ and
  $Q_{(ii)}$, and the copies of a control state $q\in Q$ are denoted by
  $q_{(i)}$ and $q_{(ii)}$. We define $Q'=Q\cup Q_{(i)}\cup Q_{(ii)}$. We
  duplicate the set of actions $A$ into two additional copies
  $A_{(i)}$ and $A_{(ii)}$. The copies of an action $a\in A$ are denoted
  by $a_{(i)}$ and $a_{(ii)}$. We introduce the set of transitions

  \begin{equation*}
    \begin{array}{l@{\;}l@{\;}l}
      T_{(i)}&=\{(p_{(i)},a_{(i)},q_{(i)}) \mid (p,a,q)\in T \land a\in A\}&{}\cup\{(\qf,a_{(i)},q_{(i)}) \mid (\qf,a,q)\in T \land a\in A\},\\
      T_{(ii)}&=\{(p_{(ii)},a_{(ii)},q_{(ii)}) \mid (p,a,q)\in T\}&{}\cup \{(\qf,a_{(ii)},q_{(ii)}) \mid (\qf,a,q)\in T\},
    \end{array}
  \end{equation*}
  where $(\az)_{(ii)}$ denotes $\az$. Observe that transitions in
  $T_{(i)}$ are not labeled by the zero-test~$\az$. The set of
  transitions of $\VV'$ is $T'=T\cup T_{(i)}\cup T_{(ii)}$. The displacement
  function $\delta'$ is defined by $\delta'(a)=(\delta(a),\delta(a))$, and
  $\delta'(a_{(i)})=\delta'(a_{(ii)})=(\delta(a),\vec{0})$ for every $a\in A$, and
  $\delta'(\az)=(\vec{0},\vec{0})$.  Now just observe that for every
  $\vec{x},\vec{y}\in\setN^d$, we have:
  
  \begin{itemize}
  \item [$(i)$] There exists a run in $\VV$ of the form
    $(\qini,\ini)\xrightarrow{~*~}(\qf,\vec{x})\xrightarrow{~w~}(q,\vec{y})$
    such that $w\in A^+$ if and only if $(q_{(i)},\vec{y},\vec{x})$ is
    reachable in $\VV'$.

  \item[$(ii)$] There exists a run in $\VV$ of the form
    $(\qini,\ini)\xrightarrow{~*~}(\qf,\vec{x})\xrightarrow{~w~}(q,\vec{y})$
    such that $w\in (A\cup\{\az\})^+$ if and only if
    $(q_{(ii)},\vec{y},\vec{x})$ is reachable in $\VV'$.
  \end{itemize}
  From Lemma~\ref{lem:repeat} we deduce that $\qf$ is a repeated
  control state in $\VV$ if and only there exists for $\VV'$ a
  reachable state of the form $((\qf)_{(i)},\vec{y},\vec{x})$ with
  $\vec{x}\leq \vec{y}$, or a reachable state of the form
  $((\qf)_{(ii)},\vec{y},\vec{x})$ with $\vec{x}\leq_1\vec{y}$.
  
  We reduce these two problems to the reachability problem for a
  \vassz $\VV''$ obtained from $\VV'$ by adding two extra states
  $r_{(i)}$ and $r_{(ii)}$, two extra transitions
  $((\qf)_{(i)},(\vec{0},\vec{0}),r_{(i)})$ and
  $((\qf)_{(ii)},(\vec{0},\vec{0}),r_{(ii)})$, and two extra cycles on
  $r_{(i)}$ and $r_{(ii)}$ that suitably decrease the counters, in
  such a way that
  \begin{itemize}[label=--,leftmargin=*]
  \item $((\qf)_{(i)},\vec{y},\vec{x})$ with $\vec{x}\leq \vec{y}$ is
    reachable in $\VV'$ if and only if $(r_{(i)},\vec{0},\vec{0})$ is
    reachable in $\VV''$, and
  \item $((\qf)_{(ii)},\vec{y},\vec{x})$ with $\vec{x}\leq_1\vec{y}$ is
    reachable in $\VV'$ if an only if $(r_{(ii)},\vec{0},\vec{0})$ is
    reachable in $\VV''$.
  \end{itemize}
  We have reduced the repeated control state reachability problem to
  the reachability problem for \vassz, which is
  decidable~\cite{Reinhardt:08,MFCS:11}.
\end{proof}

A classical application of the decidability of the repeated control
state reachability for \vass is the decidability of LTL
model-checking, and more generally of model-checking against
$\omega$-regular specifications (it is well-known that LTL specifications
can be effectively compiled into $\omega$-regular specifications, see
\cite{Vardi&Wolper:Reasoning-about-Infinite-Computations:1994:a} for
some original results, or
\cite{Vardi:automata-theoretic-approach-linear-temporal:1996:a} for a
survey). Let us informally describe this problem
(see~\cite{Esparza:98,Blockelet&Schmitz:Model-Checking-Coverability-Graphs:2011:a}
for formal presentations). Its inputs are a $\Sigma$-labeled \vassz $\VV$
and an $\omega$-regular language $L$ over $\Sigma$. By a $\Sigma$-labeled \vassz, we
mean a \vassz $\VV$ with transition set $T$, equipped with a labeling
function $\ell:T\to\Sigma$. The \emph{trace} of an infinite run of $\VV$ is the
infinite word over $\Sigma$ obtained as the image under $\ell$ of the run. The
question is whether all traces of \VV belong to $L$.

\smallskip For \vass, the standard technique to solve this problem is
to build the product $\VV\times\mathcal A$ of the \vass~$\VV$ with a B\"uchi
automaton $\mathcal A$ recognizing $L$, synchronized on $\Sigma$. The
problem then reduces to the repeated control state reachability in
$\VV\times\mathcal A$, which is a \vass. This also works in our case, since
the class of \vassz is closed under direct product with a finite-state
automaton. We deduce the following statement.

\begin{theorem}
  \label{thm:ltl-reduce}
  Model-checking a labeled vector addition system with states and one
  zero-test against an $\omega$-regular property (and in particular against
  an LTL specification) is decidable.
\end{theorem}

\section{Conclusion and perspectives}
\label{sec:conclusion}

\paragraph{\bfseries Summary} Our main result is a forward algorithm,
\emph{\`a la} Karp and Miller, to compute the downward closure of the
reachability set of a non-monotonic transition system: \vasz. The
proof first goes by strengthening the decidability of the reachability
set of a \vas: we show that the \emph{limit closure} of this set is
decidable. We have then introduced new sets, sitting between the cover
and the reachability set. We have shown that the decidability of the
limit closure of the reachability set entails the decidability of
filtered covers for a usual \vas. This tool has then be used to
perform accurate macro-steps in an adapted Karp-Miller algorithm for
\vasz. Finally, we have shown how to use this result to decide place
boundedness for \vasz, as well as the repeated control state
reachability problem, and LTL model-checking.

\paragraph{\bfseries \vas vs.\ \vasz} Classical decidable problems for
VAS are still decidable for \vasz: reachability, coverability,
boundedness, place boundedness, LTL model-checking, repeated control
state state reachability. One may want to investigate which logical
properties remain decidable for \vasz (see
\eg~\cite{Blockelet&Schmitz:Model-Checking-Coverability-Graphs:2011:a}
for properties on \vas solvable using Karp-Miller trees). Note that
\vasz cannot be simulated by VAS. For instance the prefix-closure of
the language $\{a^nb^n\mid n\geq1\}^*$ can be recognized by a \vasz, but not
by a VAS~\cite{Kosaraju:73}.

\paragraph{\bfseries Complexity and dependency to the reachability problem}
Unfortunately, we cannot say anything about the complexity of the
computation of the cover for \vasz, because our proof uses the
decidability of the reachability problem for VAS as an oracle, whose
complexity is still open. Observe that, more precisely, we have used
the decidability of the reachability problem for \vas in
Section~\ref{sec:set-limits-reachable}, and this cannot be avoided to
get Theorem~\ref{thm:reclim}. However, to decide the repeated control
state reachability problem in Section~\ref{sec:rcsrp}, we have also
used a reduction to the decidability of the reachability problem, this
time for \vasz. It is not clear whether one can avoid it: we leave it
as an open problem.

\paragraph{\bfseries Future work}
Our results cannot be trivially extended to the more general class of
VAS with hierarchical zero-tests~\cite{Reinhardt:08}.  In fact, for
this class, the coverability problem and the reachability problem are
mutually reducible with immediate log-space reductions.  The
reachability problem was proved to be decidable by Reinhardt in
\cite{Reinhardt:08}. Recently, the model of VAS with hierarchical
zero-tests was proved to be equivalent to VAS with one stack encoding
bounded-index context-free languages~\cite{DBLP:conf/fsttcs/AtigG11}.
As future work, we are interested in the decidability of the
reachability problem for VAS equipped with an unrestricted stack. With
this class, it becomes possible to model client-server systems where
clients are dynamically created and destructed, identical
finite-states machines, and the server is a recursive finite-state
machine communicating by rendez-vous. The reachability problem for
this class is open. For tackling this problem, we recently
investigated a simplification of Reinhardt's decidability proof of the
reachability problem for VAS with hierarchical
zero-tests~\cite{Reinhardt:08}: for the subclass of \vasz, the first
author published a simplified proof in \cite{MFCS:11}, based on the
work of the third author~\cite{Leroux:09}.

\section*{Acknowledgements}
\label{sec:ackowledgements}

We thank the referees whose careful reading helped us to improve
the paper.

\end{document}